\RequirePackage{amsmath}
\documentclass[runningheads,a4paper]{llncs}

\usepackage[dvipsnames]{xcolor}
\usepackage{amssymb} 
\usepackage{amsfonts}

\usepackage{cite} 
\usepackage{graphicx} 
\usepackage{supertabular}
\usepackage{subcaption}
\usepackage{physics}
\usepackage{soul}
\usepackage{multirow}
\usepackage[export]{adjustbox}
\usepackage[english]{babel}
\usepackage{slashed}
\usepackage{mathtools}
\usepackage{qcircuit}
\usepackage{stmaryrd}
\usepackage{pgfplots}
\usepackage{qcircuit}
\usepackage[toc,page,titletoc]{appendix}
\usepackage{apptools}
\usepackage{thm-restate}

\usepackage{listings}
\usepackage{tikz}
\usepackage{mathdots}
\usepackage{yhmath}
\usepackage{cancel}
\usepackage{siunitx}
\usepackage{array}
\usepackage{tabularx}
\usepackage{booktabs}
\usetikzlibrary{fadings}
\usetikzlibrary{patterns}
\usetikzlibrary{shadows.blur}
\usetikzlibrary{shapes}

\usepackage{xspace}
\usepackage[colorlinks=true,linkcolor=blue]{hyperref}
\usepackage[capitalise]{cleveref}
\usepackage{caption}

\providecolor{DarkBlue}{rgb}{0,0,.545}
\providecolor{DarkGreen}{rgb}{0,.392,0}
\hypersetup{citecolor=DarkGreen}
\hypersetup{linkcolor=DarkBlue}
\hypersetup{urlcolor=DarkBlue}

\usepackage{algorithm,algorithmicx}
\usepackage[noend]{algpseudocode}

\newcommand{\p}[1]{\Pr\left[#1\right]}
\newcommand{\Hf}{\mathrm{H}}
\newcommand{\Hfv}[1]{\mathrm{H}\big(#1\big)}
\newcommand{\wt}{\mathrm{wt}}
\newcommand{\F}{\mathbb{F}}
\newcommand{\N}{\mathbb{N}}

\newcommand{\SDP}[3]{\mathcal{SD}_{#1,#2,#3}}
\newcommand{\tmo}[1]{\tilde{\mathcal{O}}\left({#1}\right)}
\newcommand{\tmt}[1]{\tilde{\Theta}\left({#1}\right)}
\newcommand{\HybridPrange}{\textsc{Hybrid-Prange}}
\newcommand{\PuncturedHybrid}{\textsc{Punctured-Hybrid}}
\newcommand{\Combination}{\textsc{Combined-Hybrid}}
\renewcommand{\vec}[1]{\mathbf{#1}}

\newcommand{\qedeq}{\tag*{$\square$}}
\crefname{supp}{Appendix}{App}
\newlength{\strutdepth}%
\newcommand{\notes}[3]{
	\noindent{%
		\color{#1}{[#3]}\color{#1}}%
	\strut\vadjust{\kern-\strutdepth%
		\vtop to \strutdepth{%
			\baselineskip\strutdepth%
			\vss\llap{{\large\color{#1}\textbf{#2}\quad\color{black}}}\null%
		}%
	}%
}
\algnewcommand{\IfThen}[2]{
\State \algorithmicif\ #1\ \algorithmicthen\ #2\ }

\newcommand{\Vset}[2]{\mathcal{B}_{#1,#2}}

\begin{document}

	\title{An Optimized Quantum Implementation of ISD on Scalable Quantum Resources}

    \titlerunning{An Optimized Implementation of Quantum ISD }
	\authorrunning{~}
	\author{
		Andre Esser \inst{1} \and
		Sergi Ramos-Calderer\inst{1}\inst{2} \and
		Emanuele Bellini\inst{1} \and
		Jos\'e I. Latorre\inst{1}\inst{2}\inst{3} \and
		Marc Manzano\inst{4}\thanks{This work was conducted while the author was affiliated with Technology Innovation Institute.}
	}
	\institute{
		Technology Innovation Institute, United Arab Emirates  \\ \email{\{andre.esser, sergi.ramos, emanuele.bellini, jose.ignacio.latorre\}@tii.ae}
		\and
		Departament de F\'isica Qu\`antica i Astrof\'isica and Institut de Ci\`encies del Cosmos, Universitat de Barcelona, Spain \and 
		Centre for Quantum Technologies, National University of Singapore, Singapore\and
		Sandbox@Alphabet, Mountain View, CA, United States\\
		\email{marc.manzano@google.com}
	}
	
	\maketitle
	\begin{abstract}
	The security of code based constructions is usually assessed by Information Set Decoding (ISD) algorithms. In the quantum setting, amplitude amplification yields an asymptotic square root gain over the classical analogue. However, it is still unclear whether a real quantum circuit could yield actual improvements or suffer an enormous overhead due to its implementation. This leads to different considerations of these quantum attacks in the security analysis of code based proposals. In this work we clarify this doubt by giving the first quantum circuit design of the fully-fledged ISD procedure, an implementation in the quantum simulation library Qibo as well as precise estimates of its complexities. We show that against common belief, Prange's ISD algorithm can be implemented rather efficiently on a quantum computer, namely with only a logarithmic overhead in circuit depth compared to a classical implementation.

    As another major contribution, we leverage the idea of classical co-processors to design hybrid classical-quantum trade-offs, that allow to tailor the necessary qubits to any available amount, while still providing quantum speedups. Interestingly, when constraining the width of the circuit instead of its depth we are able to overcome previous optimality results on constraint quantum search.
	\end{abstract}
	
	\section{Introduction}
	The growing threat to modern widespread cryptography posed by the advancing development of quantum computers has led to a focus on other hardness assumptions. One of the leading and most promising proposals for post quantum cryptography is code based cryptography. It has a long history of withstanding classical as well as quantum attacks and is considered to rely on one of the most well understood hardness assumptions. The list of the four KEM finalists of the ongoing NIST standardization process for post quantum cryptography \cite{nist} includes one code based proposal (McEliece \cite{chou2020classic}) and two more can be found on the alternate candidate list (BIKE \cite{aragon2017bike} and HQC \cite{melchor2020hamming}) .

    At the heart of all these code based constructions lies the binary decoding or \emph{syndrome decoding} problem. This problem asks to find a low Hamming weight solution $\vec e\in \F_2^n$  to the equation $H\vec e = \vec s$, where $H \in\F_2^{(n-k)\times n}$ is a random binary matrix and $\vec s\in\F_2^{n-k}$ a binary vector.
    
    The best known strategy to solve this problem is based on Information Set Decoding (ISD) \cite{prange1962use}, a technique introduced by Prange in 1962. Since then, there has been a series of works improving on his original algorithm \cite{stern1988method,dumer1991minimum,may2011decoding,becker2012decoding,may2015computing,both2018decoding}, mostly by leveraging additional memory, exploiting some meet-in-the-middle strategies. 
    
    In the quantum setting Bernstein showed how to speed up Prange's algorithm by an amplitude amplification routine~\cite{bernstein2010grover}, which results in an asymptotic square root gain over the classical running time. The translation of advanced ISD algorithm to the quantum setting \cite{kachigar2017quantum, kirshanova2018improved} yields so far only small asymptotic improvements. Further these algorithms rely on the existence of an exponential amount of quantum RAM, which is considered very unrealistic even for mid term quantum developments. Due to this fact, all code based NIST submissions exclude these algorithms when conducting their security analysis. Moreover, the McEliece submission states that "Known quantum attacks multiply the security level of both ISD and AES by an asymptotic factor $0.5 + o(1)$, but a closer look shows that the application of Grover’s method to ISD suffers much more overhead in the inner loop" \cite{chou2020classic}.
    
    So far it was unclear if such a statement is well-founded and how much overhead a quantum implementation of the procedure by Prange would really cause. In this work, we carefully design every part of Prange's algorithm as a quantum circuit, analyze its complexities and show how to incorporate the pieces in a fully-fledged quantum ISD procedure. 
    
    In our design we put a special focus on the necessary amount of qubits. Note that several prior works also focus on qubit reduction in the \emph{few qubits} or \emph{polynomial memory} setting \cite{ekeraa2017quantum,bernstein2017low,helm2020power,biasse2020trade,banegas2021concrete}, in which the quantum algorithm is limited to the use of a polynomial amount of qubits only. Prange's algorithm falls into this regime by default, since asymptotically it \emph{only} uses a polynomial amount of memory. Nevertheless, it is especially this need for memory which limits its applicability, as all code based constructions involve parity-check matrices consisting of millions of bits. Hence, we investigate different optimizations of our initial design with regards to the amount of qubits. Furthermore, we extend the \emph{few qubits} setting by developing hybrid algorithms that enable us to reduce the already polynomial demand of qubits to any available amount while still providing a quantum speedup. 
    
    In this context we leverage the idea of classical co-processors resulting in hybrid trade-off algorithms between classical-time and quantum memory (and time). The idea of such co-processors has mostly been used to parallelize quantum circuits or instantiate circuits under depth constraints, e.g. when analyzing the quantum security of schemes under the {\texttt{MAXDEPTH}} constraint specified by NIST \cite{aragon2017bike,jaques2020implementing,biasse2020framework,biasse2020trade}. Under depth constraints, Zalka \cite{zalka1999grover} showed that the optimal way to perform a quantum search is by partitioning the search space in small enough sets such that the resulting circuit only targeting one set at a time does not exceed the maximum depth. Then the search has to be applied for every set of the partition. However, this optimality result only holds under depth constraints, when instead imposing constraints on the width of the circuit, our trade-offs yield more efficient strategies.
    
    \paragraph{Our Contribution.} As a first contribution we design and analyze the full circuit performing the quantum version of Prange's algorithm. We give precise estimates for the circuit depth and width in the quantum circuit model. Our design shows that, against common belief, Prange's algorithm can be implemented rather efficiently on a quantum computer, namely with only a logarithmic overhead in the depth. Through further optimizations, our width optimized circuit only needs $(n-k)k$ bits to store and operate on the input matrix $H\in\F_2^{(n-k)\times n}$ and roughly $n-k$ ancillas. Additionally, we provide functional implementations of our circuits in the quantum simulation library Qibo \cite{efthymiou2020qibo, stavros_efthymiou_2020_3997195}, which is accessible on github \cite{qisd-code}. We also explore different optimizations regarding the circuit depth, including a quantum version of the Lee-Brickell improvement \cite{lee1988observation} and an adaptation of our circuits to benefit from quasi-cyclic structures in the BIKE / HQC case. 
    
    Our second major contribution is the design of hybrid quantum-classical trade-offs that address the practical limitation on the amount of qubits. In particular, these trade-offs enable quantum speedups for any available amount of qubits. We study the behavior of our trade-offs for various different choices of code parameters. Besides the coding-theoretic motivated settings of full and half distance decoding, this includes also the parameter choices made by the NIST PQC candidates McEliece, BIKE and HQC. Our trade-offs perform best on the BIKE and HQC schemes, which is a result of a combination of a very low error weight and a comparably low code rate used by these schemes. 
    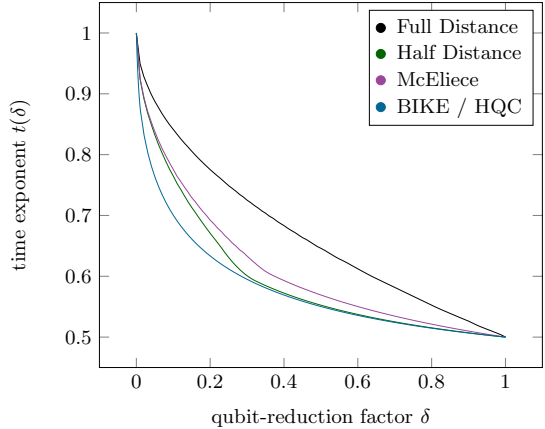
\begin{figure}[ht]
    \centering
	\scalebox{0.85}{	\begin{tikzpicture}
	\begin{axis}[
		y tick label style={
			/pgf/number format/.cd,
			fixed,
			precision=3,
			/tikz/.cd
		},
		x tick label style={
			/pgf/number format/.cd,
			fixed,
			1000 sep={},
			precision=2,
			/tikz/.cd
		},
		ytick={0.5,0.6,0.7,0.8,0.9,1},
		xlabel={qubit-reduction factor $\delta$},
		ylabel={time exponent $t(\delta)$},
		legend cell align={left},
		legend entries={~Full Distance, ~Half Distance, ~McEliece ,~BIKE / HQC}
		]
		\addlegendimage{only marks,black}
		\addlegendimage{only marks,DarkGreen}
		\addlegendimage{only marks,Purple}
		\addlegendimage{only marks,MidnightBlue}
		

		\pgfplotstableread{plots/combined_full_distance_asym.txt}
		\cTradeOff
		\addplot[color=black] table[x = X,y=Y] from \cTradeOff ;

		\pgfplotstableread{plots/combined_half_distance_asym.txt}
		\cTradeOff
		\addplot[color=DarkGreen] table[x = X,y=Y] from \cTradeOff ;
		
		\pgfplotstableread{plots/combined_mceliece_asym.txt}
		\cTradeOff
		\addplot[color=Purple] table[x = X,y=Y] from \cTradeOff ;
		
		\pgfplotstableread{plots/second_bike_asym.txt}
		\cTradeOff
		\addplot[color=MidnightBlue] table[x = X,y=Y] from \cTradeOff ;
	\end{axis}	
\end{tikzpicture}}
	\caption{Comparison of the achieved speedups of our trade-offs $t(\delta)$ (y-axis) plotted as a function of the qubit-reduction factor $\delta$ (x-axis).}
	\label{fig:intro-plot}
    \end{figure}
    
    \cref{fig:intro-plot} shows the behavior of our trade-off achieving the best results under limited width. Here, we measure the performance of the trade-offs in form of a qubit-reduction factor $\delta$ and a speedup $t(\delta)$. In comparison to an entirely quantum based computation, performed using a specific amount of qubits and taking time $T^\frac{1}{2}$, the trade-off reduces the amount of qubits by a factor of $\delta$, while maintaining a time complexity of $T^{t(\delta)}$. For instance in the BIKE and HQC setting we can reduce the amount of qubits to only 1\% ($\delta=0.01$) of an entire quantum based computation and still achieve a speedup of roughly $0.87$ compared to a classical computation.
    
	
	\smallskip
	The rest of this work is structured as follows. In \cref{sec:prelim} we set up the necessary notation, give a precise definition of the problem under consideration and elaborate on the quantum model used for our analysis. In the subsequent \cref{sec:classical-isd} we present the original algorithm by Prange. In \cref{sec:quantum-circuit} we model every step of Prange's algorithm as a quantum circuit and show how to make use of an amplitude amplification step. Finally, in \cref{sec:optimizing} we give improvements for our initial design, including a quantization of the Lee-Brickell improvement as well as our hybrid classical-quantum trade-offs.

	\section{Preliminaries}
	\label{sec:prelim}
	For two integers $a,b\in \N$ with $a\leq b$ let $[a,b]:=\{a,a+1,\ldots,b\}$. Further we write conveniently $[b]:=[1, b]$. Let $H$ be an $m\times n$ matrix and $I\subseteq [n]$, we write $H_I$ to denote the projection of $H$ onto the columns indexed by $I$. We use the same notation for vectors, so for a vector $\vec v$ of length $n$, then $\vec v_I:=(v_{i_1},v_{i_2},\ldots,v_{i_{|I|}})$, where $i_j\in I$. For a binary vector $\vec w$ we define $\wt(\vec w):=|\{i\in[n]\mid w_i=1\}|$ as the Hamming weight of $\vec w$.
	For two reals $c,d\in\mathbb{R}$ we let $\llbracket c,d\rrbracket:=\{x\in \mathbb{R} \mid c\leq x\leq d\}$ be the (including) interval of all reals between $c$ and $d$.
	
	We use standard Landau notation for complexity statements, where 
	$\tilde{\mathcal{O}}$-notation suppresses polylogarithmic factors, meaning $\tilde{\mathcal{O}}\big(f(x)\big)=\mathcal{O}\big(f(x)\log^i f(x)\big)$ for any constant $i$. Besides standard binomial coefficients, for a set $S$, we let $\binom{S}{k}$ denote the set containing all size-$k$ subsets of $S$.
	All logarithms are binary if not stated otherwise. We define $\Hf(x):=-x\log(x)-(1-x)\log(1-x)$ to be the binary entropy function and make use of the well-known approximation
	\begin{align}
	\binom{n}{k}=\tmt{2^{n\Hf\left(\frac{k}{n}\right)}},
	\label{eqn:stirling}
	\end{align}
	which can be derived from Stirling's formula.
	
	\smallskip
	\paragraph{Decoding and linear codes.} A binary linear $[n,k]$\footnote{Note that we also use this notation to indicate the set of integers between $n$ and $k$, but the concrete meaning will be clear from the context.} or $[n,k,d]$ code $\mathcal{C}$ is a $k$ dimensional subspace of $\F_2^n$ and minimum distance $d$, which is defined as the minimum Hamming weight of the elements of $\mathcal{C}$. We call $n$ the code length and $R:=\frac{k}{n}$ the code rate of $\mathcal{C}$. The code $\mathcal{C}$ can be defined via the kernel of a matrix $H\in\F_2^{(n-k)\times n}$, so that $\mathcal{C}:=\{\vec c \in \F_2^n\mid H\vec c^T=\vec 0\}$, where $H$ is called a \emph{parity-check matrix}. Note that for ease of exposition, we treat all vectors as column vectors so that we can omit vector transpositions.
	
    A given point $\vec x=\vec c+\vec e$ that differs from a codeword by an error $\vec e$ can be uniquely decoded to $\vec c$ as long as $\wt(\vec e)\leq \left\lfloor\frac{d-1}{2}\right\rfloor$. This setting, in which the error weight is bounded by half of the minimum distance, is also known as \emph{half distance} decoding. Another well-studied case upper bounds the error weight by the full minimum distance $d$ and is hence known as \emph{full distance} decoding. As the running time of decoding algorithms is solely increasing in the error weight, in those settings, we study the complexity for the case of equality to the respective upper bounds. Also in those cases we assume $d$ to meet the Gilbert-Varshamov bound \cite{gilbert1952comparison,varshamov1957estimate} for random binary linear codes, which gives $d\approx\Hf^{-1}(1-R)n$.
	
	Note that the definition of the code via its parity-check matrix allows to treat the decoding procedure independently of the specific codeword by considering the \emph{syndrome} $\vec s$ of a given faulty codeword $\vec x$, where $\vec s:=H\vec x=H(\vec c+ \vec e)=H\vec e$.
	
	Now, if one is able to recover $\vec e$ from $H$ and $\vec s$, the codeword can be recovered from $\vec x$ as $\vec c = \vec x+ \vec e$. This leads to the definition of the \emph{syndrome decoding problem}.
	
		\begin{definition}[Syndrome Decoding Problem]
	        Let ${\cal C}$ be a linear $[n,k]$ code with parity-check matrix $H\in \F_2^{(n-k)\times n}$ and constant {\em rate} $\frac k n$. For $\vec s\in\F_2^{n-k}$ and $\omega \in [n]$, the \emph{syndrome decoding problem} $\SDP{n}{k}{\omega}$ asks to find a vector $\vec e \in \F_2^n$ of weight $\wt(\vec e)=\omega$ satisfying $H\vec e = \vec s$. We call any such $\vec e$ a \emph{solution} while we refer to $(H,\vec s,\omega)$ as an \emph{instance} of the syndrome decoding problem. 
	\end{definition}

	
	\paragraph{Quantum Circuits.} Our algorithms are built in the quantum circuit model, where we assume  a certain familiarity of the reader (for an introduction see \cite{nielsen2000quantum}). The circuits are presented using general multi-qubit gates for simplicity, but we analyze their depth and complexity using their decomposition into basic implementable gates. Particularly, the decomposition of multi-controlled NOT gates into Toffoli gates is the main factor affecting depth and gate count. 
	
	A multi-controlled NOT gate with $n$ controls can be decomposed using $\mathcal{O}(n)$ regular Toffoli gates \cite{barenco1995elementary}. If $n-2$ ancilliary qubits are available the procedure can be implemented in logarithmic depth. 
	Note that we use the term circuit depth and time complexity interchangeably when analyzing our quantum circuits.
	
	\section{Prange's Information Set Decoding}
	\label{sec:classical-isd}

    Let us introduce the original ISD algorithm by Prange \cite{prange1962use}.
    Given an instance $(H,\vec s, \omega)$ of the $\SDP{n}{k}{\omega}$ Prange's algorithm starts by choosing a random set $I\subseteq [n]$ of size $n-k$ and then solves the corresponding linear system \begin{align}
    	H_I\vec e_1=\vec s
    	\label{eqn:projected-ls}
    \end{align}
    for $\vec e_1$.\footnote{Note that in \cref{alg:prange} we model $H_I$ as the first $n-k$ columns of $HP$, where $P$ is a random permutation matrix.}
    Note that the solution $\vec e_1$ of the above linear system with weight $\omega':=\wt(\vec e_1)$ can always be naively extended to a vector $\tilde{\vec e}$ of length $n$ and weight $\omega'$ satisfying $H\tilde{\vec e}=\vec s$.
    
    For the construction of $\tilde{\vec e}$ one simply sets the coordinates corresponding to the disregarded columns of $H$ to zero. Hence, if $\omega' = \omega$ the vector $\tilde{\vec{e}}$ forms a solution to the syndrome decoding problem. The algorithm now chooses random subsets $I$ until $\omega' = \omega$ holds.
    
    	\begin{algorithm}[t]
    	\begin{algorithmic}[1]
    		\Require{parity-check matrix $H\in\mathbb{F}_2^{(n-k)\times n}$, syndrome $\vec s \in \mathbb{F}_2^{n-k}$, weight $\omega\in[n]$}
    		\Ensure{error vector $\vec e$ with $\wt(\vec e)=\omega$ satisfying $H\vec e=\vec s$ }
    		\Repeat\label{line:prange-repeat}
    		\State choose random permutation matrix $P\in\mathbb{F}_2^{n\times n}$ and set $H_I\leftarrow (HP)_{[n-k]}$\label{line:prange-perm}
    		\State solve linear system $H_I\vec e_1 = \vec s$ for $\vec e_1$\label{line:prange-solve-ls}
    		\Until{$\wt(\vec e_1)=\omega$}
    		\State\Return $P(\vec e_1,0^k)$
    	\end{algorithmic}
    	\caption{\textsc{Prange}}
    	\label{alg:prange}
    	
    \end{algorithm}
    
    Let us briefly analyze when \cref{alg:prange} succeeds in finding the solution. Assuming a unique solution $\vec e$, the algorithm is successful whenever $\vec e$ projected to the coordinates given by $I$ is a solution to the linear system in \cref{eqn:projected-ls}, hence if $\vec e_1=\vec e_I$. This happens whenever $\vec e_I$ covers the full weight of $\vec e$, in which case $I$ or more precisely $[n]\setminus I$ is called an \emph{information set}. Transferred to \cref{alg:prange} this applies whenever, for the permutation chosen in line \ref{line:prange-perm}, it holds that $P^{-1}\vec e=(\vec e_1,0^k)$ for $\vec e_1 \in\F_2^{n-k}$. The probability that the permutation distributes the weight in such a way is
    \begin{align}
    q:=\p{P^{-1}\vec e=(\vec e_1,0^k)}=\frac{\binom{n-k}{\omega}}{\binom{n}{\omega}}\enspace.
    \label{eqn:prob-isd}
    \end{align}
    Hence, the expected number of tries until we draw a suitable permutation $P$ becomes $q^{-1}$ and the expected time complexity is $T=q^{-1}\cdot T_G$, where $T_G$ describes the cost for solving the linear system and performing the weight check.

    \begin{remark}
    Note that in the case of $S$ existent solutions the time complexity to retrieve a single solution with Prange's algorithm becomes $\frac{T}{S}$.
    \label{rem:isd-mult-sols}
    \end{remark}

	\section{A first design of a quantum ISD circuit}
	\label{sec:quantum-circuit}
	In this section we give an initial design for the quantum version of Prange's ISD algorithm. Our design is composed of the following three main building blocks:
	
	\begin{itemize}
		\item[1)] The creation of the uniform superposition over all size-$k$ subsets of $[n]$ (corresponding to the selection of information sets in line \ref{line:prange-perm} of \cref{alg:prange}).\\[-3mm]
		\item[2)] The Gaussian elimination step to derive the error related to a given information set (line \ref{line:prange-solve-ls} of \cref{alg:prange}).\\[-3mm]
		\item[3)] A quantum search for an information set yielding an error of the desired weight (substituting the repeat loop in line \ref{line:prange-repeat} of \cref{alg:prange}).
	\end{itemize} 
	
	Next, we give independent descriptions of our circuit designs for the different steps after which we discuss how to incorporate them in a quantum search. We provide implementations of all described circuits in the quantum simulation library Qibo with the source code accompanying this work.
	
	\subsection{Superposition over size-$k$ subsets}
	We represent a size-$k$ subset $S\subset[n]$ via a binary vector $\vec b$ of length $n$ with exactly $k$ bits set to one, where $i\in S$ iff $b_i=1$. Let $\Vset{n}{k}$ denote the set of all such binary vectors. Our circuit builds the uniform superposition over $\Vset{n}{k}$ in a bit by bit fashion. Grover and Rudolph in \cite{grover2002creating} follow a similar approach of dividing the probability space into smaller parts to provide a feasibility argument for creating quantum states following a broad class of probability distributions. However, they leave it open how to create these quantum states for concrete distributions.
	
	Our design relies on the observation that the fraction of vectors from $\Vset{n}{k}$ starting with a zero or a one respectively is known a priori and independent of subsequent bits. Therefore note that $\Vset{n}{k}$ splits into
	\begin{align}
		|\Vset{n}{k}|=\binom{n}{k}=\underbrace{~|\Vset{n-1}{k}|~}_{\textrm{elements starting with zero}}+\underbrace{~|\Vset{n-1}{k-1}|~}_{\textrm{elements starting with one}}.
		\label{eqn:recursion-superpos}
	\end{align}
	Hence, we start by rotating the first qubit such that we measure a zero with probability $a:=\frac{|\Vset{n-1}{k}|}{|\Vset{n}{k}|}=\frac{\binom{n-1}{k}}{\binom{n}{k}}$ and respectively a one with probability $b:=\frac{|\Vset{n-1}{k-1}|}{|\Vset{n}{k}|}=\frac{\binom{n-1}{k-1}}{\binom{n}{k}}$ , i.e.,
	\[
    \ket{0}\mapsto \sqrt{a}\ket{0}+\sqrt{b}\ket{1}.
	\]
    
	For the second bit we proceed similar. In the case of the first qubit being zero the remaining combinations split analogously to \cref{eqn:recursion-superpos} in $|\Vset{n-1}{k}|=|\Vset{n-2}{k}|+|\Vset{n-2}{k-1}|$, otherwise we find 
	$|\Vset{n-1}{k-1}|=|\Vset{n-2}{k-1}|+|\Vset{n-2}{k-2}|$. This allows us to rotate the second bit accordingly (this time depending on the first qubit) such that
	\begin{align*}
    \sqrt{a}\ket{00}+\sqrt{b}\ket{10}\mapsto&\sqrt{a}\left(\sqrt{\frac{|\Vset{n-2}{k}|}{|\Vset{n-1}{k}|}}\ket{00}+\sqrt{\frac{|\Vset{n-2}{k-1}|}{|\Vset{n-1}{k}|}}\ket{01}\right)\\
    &+\sqrt{b}\left(\sqrt{\frac{|\Vset{n-2}{k-1}|}{|\Vset{n-1}{k-1}|}}\ket{10}+\sqrt{\frac{|\Vset{n-2}{k-2}|}{|\Vset{n-1}{k-1}|}}\ket{11}\right).
	\end{align*}
	We now proceed analogously for the remaining bits, where each bit depends on the state of its successors. A crucial observation is that for a general position $i$, the fraction of elements having the $i$-th bit equal to zero or one  does not depend on the \emph{exact} pattern of the previous $i-1$ bits, but only on their weight. Therefore, consider a general state where the first $i-1$ qubits have already been processed and their weight is equal to $j$.
	Then all combinations for the remaining $n-(i-1)$ bits are given by $\Vset{n-(i-1)}{k-j}$. Again the number of elements starting with a zero (or one respectively), can be derived, analogously to \cref{eqn:recursion-superpos}, as
	\[
	|\Vset{n-(i-1)}{k-j}|=\underbrace{|\Vset{n-i}{k-j}|}_{\textrm{elements starting with zero}}+\underbrace{|\Vset{n-i}{k-j-1}|}_{\textrm{elements starting with one}} .
	\]
    Now, by keeping track of the weight via some auxilliary qubits we can perform the needed rotations for every bit controlled on these ancillas. \cref{alg:superposition} gives a description in pseudocode on how to construct the circuit.
	
	As in general there are $k$ possibilities for the weight of the succeeding bits and we need to process a total of $n$ bits, our circuit design achieves a depth of $\mathcal{O}(n\cdot k)$. 
	To keep track of the weight $c$ of processed bits we use $\lceil\log(k+1)\rceil$ ancillas to store the binary representation of $c$. The binary additions enlarge the circuit depth by a factor of $\mathcal{O}(\log k)$. The execution of the rotation gates controlled by the ancillary register $c$ involves multi-controlled gates with $\log k$ controls, whose decomposition contributes with an additional $\mathcal{O}(\log \log k )$ factor, if we allow for further $\log k$ ancillary qubits. In total this yields a circuit depth of $D_\mathrm{sup}=\mathcal{O}(nk\log k\log\log k)$ using $n+2\lceil\log(k +1)\rceil$ qubits to create the superposition over $\Vset{n}{k}$.

	\begin{algorithm}[t]
		\begin{algorithmic}[]
			\Require{$n,k\in\N$ with $k\leq n$, $n$ qubits $b_i$ and $\lceil \log(k+1)\rceil$ ancilla qubits (to store $c$)}
			\Ensure{Uniform superposition over $\Vset{n}{k}$, represented by the $b_i$'s}
			\State $c \leftarrow k$
			\For{$i=0$ to $n-1$}
			\For{$j=k$ down to $1$}
			\If{$c=j$}
			\State Rotate $b_i$ such that $1$ is measured with probability $\frac{|\Vset{n-i-1}{j-1}|}{|\Vset{n-i}{j}|}$
			\EndIf
			\EndFor
			\If{$b_i=1$}
			\State $c \leftarrow c-1$
			\EndIf
			\EndFor
			\State \Return $\vec b$
		\end{algorithmic}
		\caption{Generate Uniform Superposition over $\Vset{n}{k}$}
		\label{alg:superposition}
	\end{algorithm}
	In \cref{fig:superposition_circuit} we give an example of our circuit for the case of $n=5$ and $k=2$.
	
	\begin{figure}[ht]
    \centering
    \resizebox{0.7\textwidth}{!}{
    \input{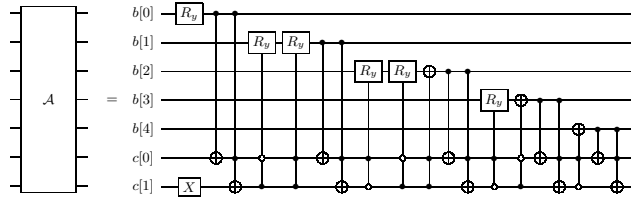}
    }

    \caption{Quantum circuit 
    generating uniform superposition over $\Vset{5}{2}$.}
    \label{fig:superposition_circuit}
    \end{figure}
	
	\subsection{Quantum Gaussian Elimination}
	
	
	Classical Gaussian elimination splits in the transformation to reduced row echelon form and finally the matrix diagonalization or solving step via back substitution. The first step performs a pivot search to enforce a non-zero diagonal entry and eliminates entries below the diagonal, while the latter eliminates entries above the diagonal.
	By performing all operations on the augmented matrix, containing the target vector as a last column, the solution to the system can be obtained from the last column of the final matrix.
	
	Our Gaussian elimination circuit mostly resembles this classical procedure but modeled only with quantum gates.
	To save on depth and width we perform row additions during the transformation to reduced row echelon form only on columns succeeding the current column. Similarly, during the back substitution step row additions are only performed on the actual solution register. The pseudocode to generate our quantum circuit is given in \cref{alg:gaussian}.
	
	
	Our circuit needs no additional qubits besides the description of the linear system and achieves a depth analogous to the classical counterpart of $\mathcal{O}(n^3)$ but involving multi-controlled gates. Our decomposition strategy for these multi-controls introduced by the pivot search (compare to line \ref{line:LS-pivot-search} in \cref{alg:gaussian}) costs an additional factor of $\mathcal{O}(\log n)$ in depth and an additional amount of $n-2$ qubits, resulting in a total depth of $D_\mathrm{gauss}=\mathcal{O}(n^3\log n)$ and a width of $W_\mathrm{gauss}=n^2+ n-2$.
    
	\begin{algorithm}[t]
		\begin{algorithmic}[1]
			\Require{matrix $A\in F_2^{n\times n}$, target vector $\vec t \in \F_2^n$ and $n(n+1)$ qubits}
			\Ensure{$\vec x \in \F_2^n$ with $A\vec x=\vec t$}
			\State Initialize qubits with augmented matrix 
			$A \leftarrow (A \mid \vec t)$
			\For{$i=1$ \textbf{to} $n-1$}
			    \For{$j=i+1$ \textbf{to} $n$}\Comment{pivot search}
			        \If{$\forall_{i\leq k<j} a_{ki}=0$}\label{line:LS-pivot-search}
			            \State add row $j$ to row $i$ starting from column $i+1$
			        \EndIf
			    \EndFor
			    \For{$j=i+1$ \textbf{to} $n$} \Comment{row reduce}
			        \If{$a_{ji}=1$}
			            \State add row $i$ to row $j$ starting from column $i+1$
			        \EndIf
			    \EndFor
			\EndFor
			\For{ $i=n$ down to $2$}\Comment{back substitution}
			    \For{ $j=i-1$ down to $1$}
			        \If{$a_{ji}=1$}
			            \State add row $i$ to row $j$ only on the last column: $a_{j(n+1)}\leftarrow a_{j(n+1)}+a_{i(n+1)}$
			        \EndIf
			    \EndFor
		    \EndFor
		    \State \Return $\vec x = (a_{1(n+1)}, \ldots, a_{n(n+1)})$\label{line:gauss-backsub-add}
		\end{algorithmic}
		\caption{Solve Linear System}
		\label{alg:gaussian}
	\end{algorithm}

	\subsection{Designing a combined circuit}
	\label{sec:combined-circuit}
	Next, we show how to combine both previously introduced building blocks to create a circuit which generates the uniform superposition over the solutions to the linear systems 
	$H_I\vec e_1=\vec s$, for $I\in \binom{[n]}{k}$. 
	
	Our combined circuit works in-place by first swapping the columns corresponding to the selected subset to the front of the matrix. Then we apply the Gaussian Elimination circuit to the first $n-k$ columns.

	\paragraph{Circuit Description.}
 First we generate the uniform superposition over $\Vset{n}{n-k}$, which determines the current selection of columns belonging to $H_I$. 
	Next we swap all columns belonging to $H_I$ to the front of the matrix (controlled on the chosen subset).
	Now, we implement the previously described Gaussian elimination circuit on the first $n-k$ columns of $H$.
	Since the final goal is to find a low weight solution it follows an accumulation of the weight of the solution in a separate register, later used by the amplitude amplification procedure. In \cref{alg:combined-circuit} we give a pseudocode description for the combined circuit generation. \cref{fig:combined} illustrates the different operators needed in order to execute the combined circuit.

	\begin{algorithm}[t]
		\begin{algorithmic}[1]
			\Require{matrix $H\in\F_2^{(n-k)\times n}$, syndrome $\vec s\in \F_2^{n-k}$, $n+(n+1)(n-k)+\lceil\log (n-k)\rceil  $ qubits}
			\Ensure{Uniform superposition over weight of all $\vec e_1$ with $H_I\vec e_1=\vec s$ for $I\in\binom{[n]}{n-k}$}
		\State Initialize qubits with $(H\mid \vec s)$
		\State Generate uniform superposition over $\Vset{n}{n-k}$ on qubits $(b_1,\ldots,b_n)$
		\For{$i = 1$ \textbf{to} $n$} \Comment{swap $H_I$ to the front}
		    \If{$b_i = 1$}
        		\For{$j = i-1$ \textbf{down to} $1$}
        		    \State swap column $j$ and $j+1$
        		\EndFor
        	\EndIf
		\EndFor
		\State Apply Gaussian elimination circuit (\cref{alg:gaussian})
		\State \Return $c \leftarrow$ weight of last column  
		\end{algorithmic}
		\caption{Combined Circuit}
		\label{alg:combined-circuit}
	  
	\end{algorithm}
	\paragraph{Circuit Complexity.} In total the amount of qubits necessary for our initial design can be summarized as 
	\begin{align*}
	W_\mathrm{combined} &= \underbrace{n}_\mathrm{Permutation} + \underbrace{(n-k)\cdot (n+1)}_\mathrm{Matrix} + \underbrace{n-k-2}_\mathrm{Auxiliary}\\
	&=(n-k+1)(n+2)-4.
	\end{align*}
	Note that all subroutines use the same auxilliary qubits as we ensure they return to the zero state after each procedure.
	
	The depth of our circuit is dominated by the circuit that swaps the columns belonging to $H_I$ to the beginning of $H$ as well as the Gaussian elimination and summarizes as 
	\begin{align}
	\begin{split}
	D_\mathrm{combined}=\mathcal{O}\big(&\underbrace{n(n-k)^2\log(n-k)}_\mathrm{Gaussian}+\underbrace{n(n-k)\log (n-k)\log\log(n-k)}_\mathrm{Permutation}\\
	&+\underbrace{n^2(n-k)}_\mathrm{Swaps}\big)=\mathcal{O}(n^3\log n)
	\end{split}
	\label{eqn:depth-combined}
	\end{align}

	\begin{figure}[t]
    \centering
    \resizebox{0.25\textwidth}{!}{
    \input{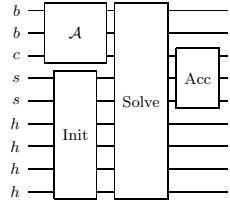}
    }

    \caption{Different parts detailed in \cref{alg:combined-circuit} as quantum operators. $\mathcal{A}$ represents the creation of superposition, \textbf{Init} initializes the registers with $(H\mid \vec s)$, \textbf{Solve} solves the linear system, \textbf{Acc} accumulates the weight of the solution.}
    \label{fig:combined}
    \end{figure}    
	
	\subsection{Amplitude Amplification}
	
	Amplitude amplification was introduced as a generalization of Grover's algorithm \cite{grover1996fast} in \cite{brassard1997exact,grover1998quantum} and analogously allows to obtain a square-root advantage over a classical search. More precisely, given a quantum operation $\mathcal{A}$, that creates a quantum state with non-zero overlap with a target state of amplitude $a$, one can amplify the probability of measuring the desired state to $\Theta(1)$ using $\mathcal{O}(1/a)$ iterations of operator $\mathcal{A}$, whereas classical sampling would require $\mathcal{O}(1/a^2)$. The amplitude amplification operator is defined as
	\begin{equation}
	    \mathcal{Q}=-\mathcal{A}S_0\mathcal{A}^\dagger S_t,
	\end{equation}
	where $S_0$ and $S_t$ are operators that flip the sign of the initial state and target state respectively. 
	The $\mathcal{Q}$ operator, when applied to the quantum state $\mathcal{A}\ket{0}$ amplifies the probability of measuring the quantum state targeted by $S_t$. Applying the operator $\lceil\pi/(4a)\rceil$ times then results in measuring one of the target states with probability close to $1$.
	

	\begin{figure}[ht]
    \centering
    \resizebox{0.6\textwidth}{!}{
    \input{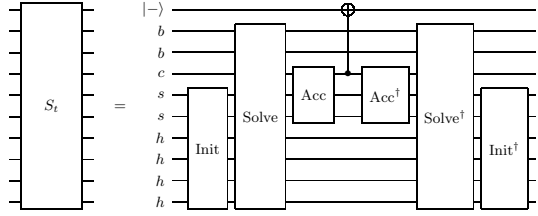}
    }
    \caption{Oracle designed to change the sign of all states with target Hamming weight after solving the linear system. }
    
    \label{fig:st}
    \end{figure}

	In our case $\mathcal{A}$ is the operator that creates the uniform superposition over all size-$(n-k)$ subsets. The oracle $S_t$ is comprised of the Gaussian elimination circuit and an ancilla initialized in the $\ket{-}$ state to flip the sign of the states with target Hamming weight. Eventually, the inverse of the Gaussian elimination circuit is applied in order to clean up the ancillary register. The shape of $S_t$ is detailed in \cref{fig:st}.  

    
    %
    
	 Summarizing, we can find a solution to the syndrome decoding problem after $\mathcal{O}(\sqrt{q^{-1}})$ applications of operator $\mathcal{Q}$, where $q$ is the proportion of suitable subsets yielding a solution among all size-$(n-k)$ subsets defined in \cref{eqn:prob-isd}.  
	 
	 The expected time complexity, becomes $\mathcal{O}(\sqrt{q^{-1}}T_{\mathcal{Q}})$, where $T_{\mathcal{Q}}$ is the cost of performing the quantum operator $\mathcal{Q}$. Note that $T_{\mathcal{Q}}=\mathcal{O}(D_\mathrm{combined})$ is dominated by the cost of the combined circuit given in \cref{eqn:depth-combined}. Hence, the depth of the full quantum ISD procedure can be summarized as
	\begin{align}
	    D_\mathrm{full}=\mathcal{O}(\sqrt{q^{-1}}D_{\mathrm{combined}}).
	    \label{eqn:depth-full}
	\end{align}
    The amplitude amplification procedure only requires a single additional qubit for the sign flip, hence the number of qubits is equal to $W_\mathrm{full}=W_\mathrm{combined}+1$.
    
	\begin{remark}
	 Note that similar to \cref{rem:isd-mult-sols}, in the case of $S$ existent solutions the proportion of subsets yielding a solution increases to $S\cdot q$. Hence $\mathcal{O}(\sqrt{(Sq)^{-1}})$ applications of the operator $Q$ suffice to find any one of these solutions.
	 \label{rem:isd-quant}
    \end{remark}	

	\section{Optimizing the quantum circuit}
	\label{sec:optimizing}
	In this section we introduce further optimizations of our initial design. First we show how to extend our circuit by the ISD improvement made by Lee-Brickell \cite{lee1988observation} and how to exploit the cyclicity in case of BIKE and HQC. Then, we show how, in the case of Prange, a preprocessing of the matrix to \emph{systematic form} and a clever adaptation of our combined circuit allows us to save $(n-k)^2$ input qubits without any increase in depth.
	\subsection{Quantum Lee-Brickell}
	\label{sec:lee-brickell}
	Lee and Brickell \cite{lee1988observation} observed that allowing for a small weight $p$ outside of the selected subset of columns can yield a polynomial runtime improvement.
    Therefore they aim for a permutation $P$, that distributes the weight on $P^{-1}\vec e=(\vec e_1,\vec e_2)\in \F_2^{n-k}\times\F_2^k$ such that 
    \begin{align}
        \wt(\vec e_1)=\omega - p \textrm{ and } \wt(\vec e_2) = p.
        \label{eqn:distribution-lee-brickell}
    \end{align}
    Again by Gaussian elimination (here modelled via the multiplication by the matrix $Q$) one transforms the identity $HP(\vec e_1,\vec e_2)=\vec s$ into
    \[
    QHP(\vec e_1,\vec e_2)=(I_{n-k}H_2)(\vec e_1,\vec e_2)=Q\vec s.
    \]
    The algorithm then enumerates in each iteration all possible candidates for $\vec e_2$ of weight $p$ and checks for every such candidate $\vec x$ if 
    \[
    \wt(\vec e_1)=\wt(Q\vec s + H_2 \vec x)=\omega -p,
    \]
    and if so outputs the solution $\vec e = P(\vec e_1,\vec x)$.
    Note that the probability for the permutation satisfying \cref{eqn:distribution-lee-brickell} improves to 
    \begin{align}
    q_\textrm{LB}=\frac{\binom{n-k}{\omega-p}\binom{k}{p}}{\binom{n}{\omega}}.
    \label{eqn:prob-lee-brickell}
    \end{align}
    On the downside in each iteration now the Gaussian elimination as well as the enumeration of all candidates must be performed which results in a total classical running time of 
    \[
    T= (q_\textrm{LB})^{-1}\cdot\left(T_\textrm{G}+\binom{k}{p}\right),
    \]
	which is optimal for a $p$ satisfying $T_\textrm{G}\approx\binom{k}{p}$, which implies constant $p$ and hence a polynomial speedup.
	
	\paragraph{Circuit Adaptation.} 
	Note that the Lee-Brickell improvement requires knowledge of $H_2$, which corresponds to the last $k$ columns of $H$ after the Gaussian elimination. Hence, we need to extend all row additions of \cref{alg:gaussian} to be also applied to these columns.\footnote{Precisely, the row additions of the back substitution step (line \ref{line:gauss-backsub-add} of \cref{alg:gaussian}) now need to be applied on the last $k+1$ columns rather than only on the last column.} Note that this does not affect the running time in $\mathcal{O}$-notation.
	
	Now, after the Gaussian elimination (compare to \cref{alg:combined-circuit}) we add a circuit that for each selection of $p$ columns of $H_2$, adds those to the last column. After the addition we check if the weight of the last column is equal to $\omega - p$ and if so set the control bit of the amplitude amplification procedure to one. After that, we reverse the addition by adding those columns again. The total depth of this enumeration circuit is $\mathcal{O}(p\binom{k}{p})$.
	
	\medskip
	The circuit depth of the Lee-Brickell quantum algorithm can be summarized similar to before in the case of Prange as 
	
	\[
	D_\textrm{LB}=\mathcal{O}\left(\sqrt{(q_\textrm{LB})^{-1}}\left(n^3\log n +p \binom{k}{p}\right)\right).
	\]
	In terms of width, the Lee-Brickell circuit has the same performance as our initial Prange design, namely
	
	\[
	W_\textrm{LB}= W_\textrm{full}= (n-k+1)(n+2)-3.
	\]
	
	The source code accompanying this work also provides an implementation of this improvement in \emph{Qibo} \cite{qisd-code}.
	 
	\subsection{The case of BIKE and HQC -- exploiting the cyclicity}
	BIKE and HQC use double-circulant codes of rate $\frac{1}{2}$, i.e., $n=2k$. For those codes, a given syndrome decoding instance $(H,\vec s,\omega)$ allows to obtain $k$ different instances $(H,\vec s_i,\omega)$\footnote{All defined on the same parity-check matrix $H$ and with the same error-weight $\omega$.}, where the solution to any one of these instances leads directly to a solution to the original instance. Sendrier has shown \cite{sendrier2011decoding} that in the case enumeration based ISD, this setting allows for a classical-speedup of $\sqrt{k}$, known as \emph{Decoding-One-Out-of-Many}. However, we observe that in the case of Prange these $k$ instances also allow for a \emph{quantum}-speedup of $\mathcal{O}(\sqrt{k})$ (corresponding to a classical speedup of order $k$). Therefore instead of performing the Gaussian elimination on the matrix $(H\mid \vec s)$ it is performed on $(H\mid \vec s_1\mid\cdots\mid\vec s_k)$. Now, whenever one of the last $k$ columns after the Gaussian elimination admits weight $\omega$ we set the sign-flip bit of the amplitude amplification procedure. Note that this change does not effect the depth of the circuit in $\mathcal{O}$-notation. It increases the cost for a row addition from $\mathcal{O}(n)$ to $\mathcal{O}(n+k)=\mathcal{O}(n)$ and the weight-accumulation has to be performed $k$ times instead of once, which is still surpassed by the cost of the Gaussian elimination. On the upside there exist a total of $k$ solutions to these $k$ instances, which according to \cref{rem:isd-quant} results in a speedup of $\sqrt{k}$ of the quantum search. In terms of width we need $(n-k)(k-1)$ more qubits to represent all $\vec s_i$.
	
	This strategy is also compatible with the Lee-Brickell improvement from the previous \cref{sec:lee-brickell}. Therefore we need to extend the enumeration part to all $\vec s_i$ increasing the cost of that step by a factor of $k$ to $\mathcal{O}(k p (n-k) \binom{k}{p})$. We summarize the width and depth of this combination in \cref{tab:circuit-complexities} at the end of \cref{sec:trade-offs}.
	
	\subsection{Saving quadratically many qubits for free}
		\begin{figure}[t]
	\centering
		\scalebox{0.5}{\begin{tikzpicture}

\node at (0.6,-2.25) {\huge$\mathbf{e}$};
\node at (6.2,0.25) {\huge$\mathbf{s}$};

\draw[opacity=0.5,densely dotted] (-4,2) -- (0.5,-1.5);
\node at (-1.8,0.25) {\huge$I_{n-k}$};
\draw[opacity=0.5,dashed] (0.5,2) node (v1) {} -- (0.5,-1.5);

\node at (2.75,0.25) {\huge$H'$};

\draw  (-4,2) rectangle (5.45,-1.5);
\draw  (-4,-2) rectangle (5.45,-2.5);
\draw[]  (5.95,2) rectangle (6.45,-1.5);
\node at (-5,0.25) {\huge$H$};

\end{tikzpicture}}
		\caption{Problem shape for input matrix in systematic form.}
		\label{fig:systematic-form}
	\end{figure}
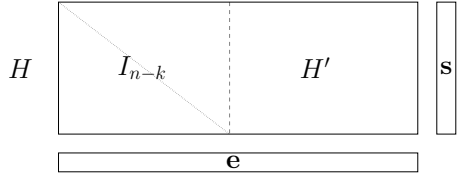
	\label{sec:width-improvement}
    In the following we apply a Gaussian elimination to the first $n-k$ columns of $H$ resulting in a problem shape as shown in \cref{fig:systematic-form}, also known as \emph{systematic form}.\footnote{If the first $n-k$ columns do not form a matrix of full rank we permute the columns accordingly.}
	We now describe how to adapt our circuit to only require the matrix $H'$ as well as the corresponding syndrome as an input, reducing the amount of qubits by $(n-k)^2$.
	
	In a classical setting, preprocessing the matrix to systematic form allows to reduce the number of operations needed for subsequent Gaussian eliminations~\cite{bernstein2008attacking}. This holds as long as there is a non-empty intersection between the new subset of columns and the first $n-k$ columns, which already form the identity matrix.

	Inspired by the classical time complexity improvement, we implement for columns from the identity part that belong to the selected subset only a corresponding row-swap. 
	After the necessary row-swaps are performed, all columns of $H'$ belonging to the corresponding subset are swapped to the \emph{back}. Subsequently we perform the Gaussian elimination only on the last columns of $H'$ that belong to the current selection. This procedure is depicted in \cref{fig:saving-qubits}, which shows the state of the matrix after all three operations have been performed for the chosen subset. Note that the first $n-k$ columns only serve an illustrative purpose and are not part of the input.
    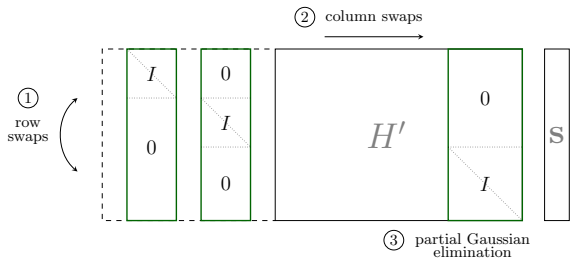
\begin{figure}[t]
	\centering
		\scalebox{0.65}{\usetikzlibrary{patterns}

\begin{tikzpicture}

\pgfdeclarepatternformonly{my crosshatch dots}{\pgfqpoint{-1pt}{-1pt}}{\pgfqpoint{5pt}{5pt}}{\pgfqpoint{6pt}{6pt}}%
{
    \pgfpathcircle{\pgfqpoint{0pt}{0pt}}{.5pt}
    \pgfpathcircle{\pgfqpoint{3pt}{3pt}}{.5pt}
    \pgfusepath{fill}
}

\node[opacity=0.5] at (6.2,0.25) {\huge$\mathbf{s}$};

\node[opacity=0.5] at (2.75,0.25) {\huge$H'$};

\draw  (0.5,2) node (v1) {} rectangle (5.5,-1.5);

\draw[]  (5.95,2) rectangle (6.45,-1.5);

\draw[dashed] (v1) -- (-3,2) -- (-3,-1.5) -- (0.5,-1.5);

\draw[thick,DarkGreen]   (-2.5,2) node (v2) {} rectangle (-1.5,-1.5);
\draw[thick,DarkGreen]   (-1,2) rectangle (0,-1.5);
\draw[thick,DarkGreen]  (4,2) rectangle (5.5,-1.5);
\draw[opacity=0.5,densely dotted]  (-2.5,2) -- (-1.5,1);
\draw[opacity=0.5,densely dotted]  (-2.5,1) -- (-1.5,1);
\draw[opacity=0.5,densely dotted]  (-1,1) -- (0,0);
\draw[opacity=0.5,densely dotted]  (-1,0) -- (0,0);
\draw[opacity=0.5,densely dotted] (4,0) -- (5.5,-1.5);
\draw[opacity=0.5,densely dotted] (4,0) -- (5.5,0);
\draw[opacity=0.5,densely dotted]  (-1,1) -- (0,1);
\node at (-2,1.5) {\large$I$};
\node at (-2,0) {\large$0$};
\node at (-0.5,0.5) {\large$I$};
\node at (-0.5,-0.75) {\large$0$};
\node at (-0.5,1.5) {\large$0$};
\node at (4.75,-0.75) {\large$I$};
\node at (4.75,1) {\large$0$};

\draw[<->, >=stealth] (-3.5,-0.5) to[in=220,out=140] (-3.5,1);
\node[align=center] at (-4.5,0.3) {row \\[-0.1cm]swaps};
\draw  (-4.5,1) ellipse (0.2 and 0.2) node {1};
\draw  (1.1,2.65) ellipse (0.2 and 0.2) node {2};
\draw[->, >=stealth](1.5,2.25) -- (3.5,2.25);
\node at (2.5,2.65) {column swaps};
\draw  (2.9,-1.9) ellipse (0.2 and 0.2) node {3};
\node[align=center] at (4.5,-2) {partial Gaussian\\[-0.1cm] elimination};
\end{tikzpicture}}
		\caption{Procedure to perform quantum version of Prange without first $n-k$ columns as input. Colored framed parts indicate columns belonging to the current selected subset.}
		\label{fig:saving-qubits}
	\end{figure}
	
	\paragraph{Circuit Adaptation}. In the following we assume the input matrix to be in \emph{systematic form}. Now, whenever the selected subset includes any column $\vec h_j$ for $j\leq n-k$ of $H$, a single row swap is sufficient to obtain the desired unit vector for this column. 
    This implies that all operations resulting from column $\vec h_j$ of $H$, where $j\leq n-k$, being part of the selected subset as column $i$ are fully determined by $j$ and $i$. The necessary operations can, hence, be embedded into the circuit directly without the need of the first $n-k$ columns as an input. However, this design requires to implement the procedure for every combination of $j$ and $i$, which are $\mathcal{O}(n^2)$ possibilities. Furthermore, each swap comes at a cost of $\mathcal{O}(n)$. 
	
	Now, it follows the part of the circuit swapping all columns of $H'$ belonging to the respective subset to the back. During these swaps we craft an ancillary state $x=\ket{0^{k-r}1^r}$, where $r$ is the number of selected columns from $H'$. This state allows us to perform the Gaussian elimination only on the last $r$ columns of $H'$ (by controlling all operations that depend on column $j$ of $H'$ on bit $x_j$). Since we do not know a priori how many columns that might be, we start the Gaussian elimination, contrary to the description before, from the back, where it always starts with the last column.
	
	In total, our modified circuit then needs an amount of qubits equal to
    \begin{align}
	W_\mathrm{improved} &= \underbrace{n}_\mathrm{Permutation} + \underbrace{(n-k)\cdot (k+1)}_\mathrm{Matrix} + \underbrace{k + n -k- 2}_\mathrm{Auxiliary} + \underbrace{1}_\mathrm{Sign-flip}\nonumber\\
	&=(n-k+2)(k+3)-7,
	\label{eqn:reduced-qubits}
	\end{align}
	since we save the identity part, but need an additional $k$ ancillas for representing $x$. The depth is still dominated by the Gaussian elimination on the last (possibly) $n-k$ columns and, hence, stays as in the inital design in \cref{eqn:depth-full}. We again provide an implementation of this design with the source code belonging to this work. A pseudocode description of the adapted combined circuit can be found in \cref{app:width-reduce}.
	
	\section{Classical-time quantum-memory trade-offs}
	\label{sec:trade-offs}
	Despite our optimization, the quadratic amount of qubits required for the representation of the input matrix is still the limiting factor with respect to real quantum implementations. 
	
	We overcome this issue by introducing hybrid trade-offs between classical-time and quantum-memory for our ISD circuit, allowing for an adaptive scaling of the algorithm to the available amount of qubits. Our trade-offs divide in a classical and quantum computation part, where a decrease of the amount of qubits comes at the cost of an increased classical running time. 
	Since this increase in running time is of exponential nature we neglect the polynomial factors of the implementation by switching to $\tilde{\mathcal{O}}$-notation.
	Our trade-offs allow for a smooth interpolation between purely classical computations at a running time of 
	\begin{align}
	T_\mathrm{C}:=\tmo{\frac{\binom{n}{\omega}}{\binom{n-k}{\omega}}},
	\label{eqn:classical-complexity}
	\end{align}
	(compare to the analysis in \cref{sec:classical-isd}) and a purely quantum based computation taking time $\sqrt{T_\mathrm{C}}$, as given in \cref{eqn:depth-full}. We interpolate between both complexities using a qubit reduction factor $\delta$, where a fully classical computation corresponds to $\delta=0$ and an entirely quantum based execution implies $\delta=1$. For each trade-off we then state the running time for a given reduction factor $\delta$ as $t(\delta)\in\llbracket 0.5,1\rrbracket$, meaning that a reduction of the amount of qubits by a factor of $\delta$ implies a total running time of $\left(T_\mathrm{C}\right)^{t(\delta)}$.
	
	We start at first with a quite straightforward trade-off, which already achieves a better than linear dependence between $\delta$ and $t(\delta)$. This first trade-off also achieves good results for concrete medium sized parameters. After that, we present a second trade-off which asymptotically outperforms the first one. However, for concrete parameters in medium scale both trade-offs remain superior to each other for certain values of $\delta$. Finally, we show how to combine both trade-offs to obtain an improved version. For large reduction factors, meaning close to one, this combination obtains the minimum of both previous trade-offs while for small reduction factors, which are most important when aiming at near future realizations, an improved running time is achieved in most settings. 
	
	We also provide an implementation of our classical co-processor in Sage that invoke the quantum circuit (implemented in the quantum simulation library Qibo) on the respective instances. 
	\begin{figure}[t]
	\centering
	
		\scalebox{0.6}{\begin{tikzpicture}
\node at (-5,-2.25) {\huge$\mathbf{e}$};
\node at (6.2,0.5) {\huge$\mathbf{s}$};

\draw[opacity=0.5,densely dotted] (-4,2) -- (0.5,-1.5);
\node at (-2.05,0.5) {\huge$I_{n-k}$};
\draw[opacity=0.5,dashed] (0.5,2) node (v1) {} -- (0.5,-1.5);

\node at (2.75,0.5) {\huge$H'$};
\fill[pattern=north west lines]  (-4,-2) rectangle (3.8,-2.5);
\fill[color=white,opacity=0.5]  (-4,-2) rectangle (3.8,-2.5);
\draw (3.8,-2) -- (3.8,-2.5);

\node at (4.65,-2.25) {\Large$\mathbf{0}$};
\fill[opacity=0.5,pattern=crosshatch]  (3.8,2)  rectangle (5.45,-1.5);
\fill[opacity=0.65,color=white]  (3.8,2) rectangle (5.45,-1.5);

\draw  (-4,2) rectangle (5.45,-1.5);
\draw  (-4,-2) rectangle (5.45,-2.5);
\draw[thick]  (5.95,2) rectangle (6.45,-1.5);
\node at (-5,0.5) {\huge$H$};

\draw (3.8,2) -- (3.8,-1.5);

\draw (-4,-2.75) -- (-4,-2.95);
\draw (-4,-2.85) -- (-0.85,-2.85);

\draw (5.05,-2.85) -- (5.45,-2.85);
\draw (5.45,-2.75) -- (5.45,-2.95);
\node at (0.05,-2.85) {$(1-\alpha) n$};

\draw (1,-2.85) -- (3.75,-2.85);
\draw (3.75,-2.75) -- (3.75,-2.95);
\draw (3.85,-2.75) -- (3.85,-2.95);
\draw (3.85,-2.85) -- (4.25,-2.85);
\node at (4.65,-2.85) {$\alpha n$};
\draw[thick]  (v1) rectangle (3.8,-1.5);
\end{tikzpicture}}
		\caption{Parity-check matrix where $\alpha n$ zero positions of $\vec e$ are guessed. Striped region of $\vec e$ indicates parts containing weight, crosshatched columns of $H'$ do not affect $\vec s$. Framed parts are used as input to the quantum algorithm.}
		\label{fig:first-tradeoff}
	\end{figure}
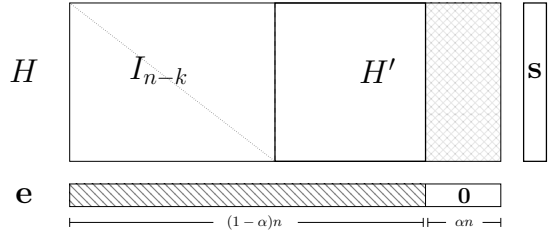
	\subsection{A hybrid version of Prange}
	Our first trade-off is a hybrid version of Prange's original algorithm. In his original algorithm $k$ zero positions of $\vec e$ are guessed and then the linear system corresponding to the non-zero positions is solved in polynomial time. In our hybrid version the classical part consists in guessing $\alpha n \leq k$ zero coordinates of $\vec e$, which allows to shorten the code and, hence, reduce the problem to an $[(1-\alpha) n,k-\alpha n]$ code, while the error weight remains unchanged (compare to \cref{fig:first-tradeoff}). This reduced instance is then solved with our previously constructed quantum circuit. Should the quantum computation not result in an actual solution, the initial guess of zero coordinates was incorrect and we proceed with a new guess.

	\cref{alg:hybrid-prange} gives a pseudocode description of our hybrid Prange variant.
	\begin{algorithm}
	\begin{algorithmic}[1]
		\Require{parity-check matrix $H\in\mathbb{F}_2^{(n-k)\times n}$, syndrome $\vec s \in \mathbb{F}_2^{n-k}$, weight $\omega\in[n]$, qubit reduction factor $\delta\in \llbracket0,1\rrbracket$}
		\Ensure{error vector $\vec e$ with $\wt(\vec e)=\omega$ satisfying $H\vec e=\vec s$ }
	
	\State $\alpha:= (1-\delta)\frac{k}{n}$
	\Repeat
	\State choose random permutation matrix $P\in\mathbb{F}_2^{n\times n}$ and set $\tilde{H}\leftarrow HP$
	\State solve instance $(\tilde{H}_{[(1-\alpha)n]},\vec s,\omega)$ via quantum algorithm returning $\vec e_1 \in \mathbb{F}_2^{(1-\alpha)n}$ \label{line:hp-quantum}
	\State $\vec e \leftarrow P(\vec e_1,0^{\alpha n})$ 
	\Until{$H\vec e = \vec s$}
	\State\Return $\vec e$
	\end{algorithmic}
	\caption{\HybridPrange}
	\label{alg:hybrid-prange}
    
	\end{algorithm}
	
	\begin{theorem}[Hybrid Prange]
	Let $n\in\mathbb{N}$, $\omega=\tau n$ and $k=Rn$ for $\tau,R\in \llbracket0,1\rrbracket$ and let $T_\mathrm{C}$ be as defined in \cref{eqn:classical-complexity}. Then for any qubit reduction factor $\delta \in \llbracket0,1\rrbracket$ \cref{alg:hybrid-prange} solves the $\SDP{n}{k}{\omega}$ problem in time $(T_\mathrm{C})^{t(\delta)}$ using $\delta (1-R)Rn^2$ qubits for the matrix representation, where 
	\[
	t(\delta)=1-\frac{\frac{1}{2}\left((1-\alpha)\Hfv{\frac{\tau}{1-\alpha}}-(1-R)\Hfv{\frac{\tau}{1-R}}\right)}{\Hf(\tau)-(1-R)\Hfv{\frac{\tau}{1-R}}},
	\label{thm:first-tradeoff}
	\]
	for $\alpha = (1-\delta)R$.
	\end{theorem}
	\begin{proof}
	Assume that the permutation $P$ distributes the error such that 
	\begin{align}
	    P^{-1}\vec e = (\vec e_1,0^{\alpha n}),
	    \label{eqn:hp-desired-dist}
	\end{align} for $\alpha$ as defined in \cref{alg:hybrid-prange}. Then it follows, that $\vec e_1$ is a solution to syndrome decoding instance $((HP)_{[(1-\alpha)n]},\vec s,\omega)$. By the correctness of our quantum circuit the solution $\vec e_1$ is returned in line \ref{line:hp-quantum} and finally $\vec e=P(\vec e_1,0^{\alpha n})$ is recovered.
	
	Next let us analyze the running time of the algorithm. The probability of a random permutation distributing the error weight as given in \cref{eqn:hp-desired-dist} is
	\[
	q_\mathrm{C}:=\mathrm{Pr}\left[P^{-1}\vec e=(\vec e_1,0^{\alpha n})\right]=\frac{\binom{(1-\alpha)n}{\omega}}{\binom{n}{\omega}}.
	\]
	Hence, we expect that after $q_\mathrm{C}^{-1}$ random permutations one of them induces the desired weight-distribution. The asymptotic time complexity for the execution of the quantum circuit to solve the corresponding $\SDP{(1-\alpha)n}{(R-\alpha)n}{\omega}$ problem can be derived from \cref{eqn:depth-full} as
	\[
	    T_\mathrm{Q} = \tmo{\sqrt{\frac{\binom{(1-\alpha)n}{\omega}}{\binom{(1-R)n}{\omega}}}}.
	\]
	
	Since for each classically chosen permutation we need to execute our quantum circuit the total running time becomes
	\[
	 T=q_\mathrm{C}^{-1}\cdot T_\mathrm{Q}=\tmo{\frac{\binom{n}{\omega}}{\sqrt{\binom{(1-\alpha) n}{\omega}\binom{(1-R)n}{\omega}}}}.
	\]
	
	Now let us determine $t(\delta):=\frac{\log T}{\log T_\mathrm{C}}$. First observe that $T=\frac{T_\mathrm{C}}{T_\mathrm{Q}}$, which can be rewritten as
	\begin{alignat*}{2}
	&&\log T_\mathrm{C}-\log T_\mathrm{Q} &= \log T\\
	\Leftrightarrow~&& 1-\frac{\log T_\mathrm{Q}}{\log T_\mathrm{C}}&=\frac{\log T}{\log T_\mathrm{C}}=:t(\delta).
	\end{alignat*}
	An approximation of $T_\mathrm{Q}$ and $T_\mathrm{C}$ via the approximation for binomial coefficients given in \cref{eqn:stirling} together with $\omega:=\tau n$ and $k:=Rn$ then yields
	\[
	t(\delta) = 1-\frac{\frac{1}{2}\left((1-\alpha)\Hfv{\frac{\tau}{1-\alpha}}-(1-R)\Hfv{\frac{\tau}{1-R}}\right)}{\Hf(\tau)-(1-R)\Hfv{\frac{\tau}{1-R}}},
	\]
	as claimed. Note that the input matrix of an $[(1-\alpha)n,(R-\alpha) n]$-code requires $(1-R)(R-\alpha) n^2$ qubits for the matrix representation (compare to \cref{eqn:reduced-qubits}). Hence, by setting $\alpha=(1-\delta)R$ we obtain a qubit reduction by
	\[
	    \frac{(1-R)(R-\alpha) n^2}{(1-R)Rn^2}=\frac{R-(1-\delta)R}{R}=\delta.
	    \qedeq
	\]
	\end{proof}

	Note that in the case of a sublinear error-weight, which is e.g. the case for the McEliece, BIKE and HQC crypto systems, $T_\mathrm{C}$ can be expressed as
	\begin{align}
	   T_\mathrm{C}=\tmo{\frac{\binom{n}{\omega}}{\binom{(1-R)n}{\omega}}}=\tmo{(1-R)^{-\omega}},
	   \label{eqn:sublinear-error}
	\end{align}
	as shown in \cite{torres2016analysis}.
	
	This allows us to simplify the statement of \cref{thm:first-tradeoff} in the following corollary.
	
	\begin{restatable}[Hybrid Prange for sublinear error weight]{corollary}{corhybridprange}
	Let all parameters be as in \cref{thm:first-tradeoff}. For $\tau=o(1)$, we have
	\[
	t(\delta)=\frac{1}{2}\cdot\left(1+\frac{\log(1-(1-\delta)R)}{\log(1-R)}\right).
	\]
	\label{cor:hybrid-prange}
	\end{restatable}
	\begin{proof}
	First we approximate $T_\mathrm{Q}$ similar to $T_\mathrm{C}$ in \cref{eqn:sublinear-error} as
	\[
	 T_\mathrm{Q}=\tmo{\sqrt{\frac{\binom{(1-\alpha)n}{\omega}}{\binom{(1-R)n}{\omega}}}}=\tmo{\left(\frac{1-\alpha}{1-R}\right)^{\frac{\omega}{2}}}.
	\]
	Now we can derive the statement of the corollary as
	\begin{align*}
	t(\delta)&=1-\frac{\log T_\mathrm{Q}}{\log T_\mathrm{C}}= 1- \frac{\frac{\omega}{2}\left(\log(1-\alpha)-\log(1-R)\right)}{-\omega\log(1-R)}\\
	&= \frac{1}{2}\cdot\left(1+\frac{\log(1-(1-\delta)R)}{\log(1-R)}\right).
	\qedeq
	\end{align*}
	\end{proof}

	\cref{fig:first-tradeoff-plots} visualizes the relation between the qubit-reduction factor and the speedup for different choices of the code- and error-rate. We compare the full distance decoding setting with worst-case rate $R=0.5$ and, hence, $\tau=\Hf^{-1}(R)\approx0.11$ and the half distance case with $\tau=\frac{\Hf^{-1}(R)}{2}$ to the code parameters of the McEliece scheme, which are $R=0.8$ and $\tau=o(1)$, and the parameters of the BIKE and HQC schemes, which are specified as $R=0.5$ and $\tau=o(1)$. Additionally, we give comparisons to higher code- and error-rates. It can be observed that the best results are obtained for high rates, where the code-rate is the more significant factor, which lies in favour to mounting an attack against codes using McEliece parameters. Note that especially for a rate close to $0.5$ the trade-off is very insensitive to changes in the error-rate, such that the behaviour for the settings of full and half distance as well as BIKE and HQC are almost identical, hence, we only included the full distance case for the sake of clarity.
	
	To give a concrete example, our \HybridPrange{} algorithm allows for a reduction of the necessary qubits by 80\% (corresponding to $\delta=0.2$), while still achieving a speedup of $t(\delta)\approx0.82$ in the McEliece setting.
	
	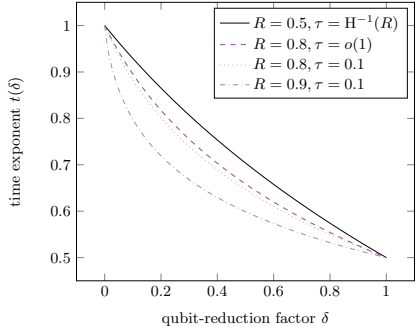
\begin{figure}[t]
	    \centering
    	\resizebox{0.45\textwidth}{!}{	\begin{tikzpicture}
	\begin{axis}[
		y tick label style={
			/pgf/number format/.cd,
			fixed,
			precision=3,
			/tikz/.cd
		},
		x tick label style={
			/pgf/number format/.cd,
			fixed,
			1000 sep={},
			precision=2,
			/tikz/.cd
		},
		ytick={0.5,0.6,0.7,0.8,0.9,1},
		xlabel={qubit-reduction factor $\delta$},
		ylabel={time exponent $t(\delta)$},
		legend cell align={left},
		legend entries={{$R=0.5, \tau=\mathrm{H}^{-1}(R)$},{$R=0.8, \tau=o(1)$},{$R=0.8, \tau=0.1$},{$R=0.9, \tau=0.1$}}
		]

		\pgfplotstableread{plots/first_r=05_full_distance.txt}
		\cTradeOff
		\addplot[color=black] table[x = X,y=Y] from \cTradeOff ;
		
		\pgfplotstableread{plots/first_r=08_small_w.txt}
		\cTradeOff
		\addplot[color=Purple,dashed] table[x = X,y=Y] from \cTradeOff ;
		
		\pgfplotstableread{plots/first_r=08_w=01.txt}
		\cTradeOff
		\addplot[color=Gray,dotted] table[x = X,y=Y] from \cTradeOff ;
		
		\pgfplotstableread{plots/first_r=09_w=01.txt}
		\cTradeOff
		\addplot[color=Gray,dashdotted] table[x = X,y=Y] from \cTradeOff ;
	\end{axis}	
	
\end{tikzpicture}}
    	\caption{Time exponent (y-axis) achieved by \cref{thm:first-tradeoff} plotted as a function of the qubit-reduction factor $\delta$ (x-axis).} 
    	\label{fig:first-tradeoff-plots}
    \end{figure}

	\subsection{Puncturing the code}
	While our \HybridPrange{} decreases the amount of necessary qubits by shortening the code, our second trade-off instead aims at puncturing the code. In a nutshell we consider only $(1-\beta)n-k$ parity-check equations, rather than all $n-k$, which is equivalent to omitting $\beta n$ rows of the parity-check matrix. The subsequently applied quantum circuit, hence, needs fewer qubits to represent matrix and syndrome. The advantage over \HybridPrange{} partly comes form the fact that each row saves $n$ instead of only $n-k$ bits. Also the generated classical overhead is significantly smaller. This variant has similarities with the  Canteaut-Chabaud improvement \cite{canteaut1998new}. Here only a certain amount of columns (originally only one) of the identity part are exchanged in each iteration rather than drawing a completely new permutation. In our case we fix $\beta n$ columns of the permutation classically and then search for the remaining $n-k-\beta n$ quantumly. In addition we introduce a different weight distribution on the fixed columns, which does not yield improvements in a purely classical setting.

    We again start with a parity-check matrix $H$ in systematic form. Now consider the projection of $H$ onto its first $n-k-\beta n$ rows, we call the resulting matrix $\tilde{H}$. Clearly, a solution $\vec e$ to the instance $(H,\vec s,\omega)$ is still a solution to the instance $(\tilde{H},\vec s_{[n-k-\beta n]}, \omega)$. Moreover, the matrix $\tilde{H}$ includes $\beta n$ zero columns, which can safely be removed (compare to \cref{fig:second-tradeoff}). This results in a matrix $\tilde{H}'=(I_{n-k-\beta n}\mid H')\in\F_2^{(n-k-\beta n)\times (1-\beta)n}$ corresponding to an $[(1-\beta) n,k]$ code. Still, by removing the corresponding coordinates from $\vec e$ we obtain a solution $\vec e'$ to the instance $(\tilde{H}',\vec s_{[n-k-\beta n]}, \omega-p)$, where $p:=\wt(\vec e_{[n-k-\beta n +1, n-k]})$ is the weight of coordinates removed from $\vec e$.
    Eventually, once $\vec e'$ is recovered we can obtain $\vec e$ in polynomial time by solving the respective linear system. 
    
    A crucial observation is that disregarding $\beta n$ parity-check equations could lead to the existence of multiple solutions to the reduced instance, i.e. multiple $\vec e'$ satisfying $\tilde{H}'\vec e' = \vec s_{[n-k-\beta n]}$ but yielding an $\vec e$ with $\wt(\vec e)>\omega$. We can control this amount of solutions by shifting more weight onto the removed coordinates. Also our algorithm compensates for multiple solutions by recovering all solutions to the reduced instance by repeated executions of the quantum circuit. A pseudocode description of this trade-off is given in \cref{alg:punctured-hybrid}.

	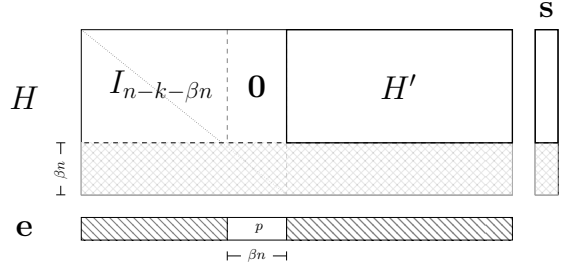
\begin{figure}[t]
	\centering
	\scalebox{0.6}{\begin{tikzpicture}

\node at (-5.25,-2.4) {\huge$\mathbf{e}$};
\node at (6.2,2.45) {\huge$\mathbf{s}$};

\draw[opacity=0.5,densely dotted] (-4,2) -- (0.5,-1.65);
\node at (-2.25,0.75) {\huge$I_{n-k-\beta n}$};
\draw[opacity=0.5,dashed] (0.5,2) node (v1) {} -- (0.5,-1.65);

\node at (2.95,0.75) {\huge$H'$};
\fill[pattern=north west lines]  (5.45,-2.15) rectangle (0.5,-2.65);
\fill[color=white,opacity=0.5]  (5.45,-2.15) rectangle (0.5,-2.65);
\draw (0.5,-2.15) -- (0.5,-2.65);

\fill[pattern=north west lines]  (5.45,-2.15) rectangle (0.5,-2.65);
\fill[color=white,opacity=0.5]  (5.45,-2.15) rectangle (0.5,-2.65);
\draw (0.5,-2.15) -- (0.5,-2.65);

\fill[pattern=north west lines]  (-4,-2.15) rectangle (-0.8,-2.65);
\fill[color=white,opacity=0.5]  (-4,-2.15) rectangle (-0.8,-2.65);
\draw (-0.8,-2.15) -- (-0.8,-2.65);

\draw  (-4,2) rectangle (5.45,-1.65);
\draw  (-4,-2.15) rectangle (5.45,-2.65);
\draw  (5.95,2) node (v2) {} rectangle (6.45,-1.65);
\node at (-5.25,0.5) {\huge$H$};

\fill[opacity=0.5,pattern=crosshatch]  (5.95,-0.5)  rectangle (6.45,-1.65);
\fill[opacity=0.65,color=white]  (5.95,-0.5) rectangle (6.45,-1.65);
\draw[dashed] (6.45,-0.5) -- (5.95,-0.5);

\fill[opacity=0.5,pattern=crosshatch]  (-4,-0.5)  rectangle (5.45,-1.65);
\fill[opacity=0.65,color=white]  (-4,-0.5) rectangle (5.45,-1.65);
\draw[dashed] (5.45,-0.5) -- (-4,-0.5);
\draw (-0.8,-2.9) -- (-0.8,-3.1);
\draw (-0.8,-3) -- (-0.5,-3);

\draw (-4.45,-0.75) -- (-4.45,-0.5);
\draw (-4.55,-0.5) -- (-4.35,-0.5);
\draw (-4.45,-1.45) -- (-4.45,-1.65);
\draw (-4.35,-1.65) -- (-4.55,-1.65);

\node at (-0.15,-3) {$\beta n$};

\draw (0.2,-3) -- (0.5,-3);
\draw (0.5,-2.9) -- (0.5,-3.1);

\node[rotate=90] at (-4.45,-1.1) {$\beta n$};
\draw[opacity=0.5,dashed] (-0.8,2) -- (-0.8,-0.5);
\node at (-0.15,0.75) {\huge$\mathbf{0}$};
\node at (-0.1,-2.4) {$p$};
\draw[fill opacity=0.04,thick]  (v1) rectangle (5.45,-0.5);
\draw[thick]  (v2) rectangle (6.45,-0.5);
\end{tikzpicture}}
	\caption{Parity-check matrix where $\beta n$ rows are omitted and $\vec e$ contains weight $p$ on $\beta n$ coordinates. Framed parts are used as input to the quantum algorithm.}
	\label{fig:second-tradeoff}
	\end{figure}

	In the following theorem we first state the time complexity of \cref{alg:punctured-hybrid} in dependence on the qubit reduction factor $\delta$. After this we derive the speedup $t(\delta)$ in a separate corollary.

	\begin{algorithm}[ht]
	\begin{algorithmic}[1]
		\Require{parity-check matrix $H\in\mathbb{F}_2^{(n-k)\times n}$, syndrome $\vec s \in \mathbb{F}_2^{n-k}$, weight $\omega\in[n]$, qubit reduction factor $\delta\in \llbracket0,1\rrbracket$}
		\Ensure{error vector $\vec e$ with $\wt(\vec e)=\omega$ satisfying $H\vec e=\vec s$ }
	
	\State choose $p$ accordingly
	\State $\beta:= (1-\delta)(1-\frac{k}{n})$, $S:=\frac{\binom{(1-\beta)n}{\omega-p}}{2^{(1-\beta)n-k}}$
	\Repeat
	\State choose random permutation matrix $P\in\mathbb{F}_2^{n\times n}$ and set $\tilde{H}\leftarrow HP$
	\State transform $\tilde{H}$ to systematic form, $\tilde{H}=\begin{pmatrix}I_{n-k-\beta n} & \vec 0 & H_1'\\
	\vec 0 & I_{\beta n}&H_2'\end{pmatrix}$ with syndrome $\tilde{\vec s}$ \label{line:RR-systematic}
	\State $\tilde{H}'\leftarrow (I_{n-k-\beta n}\mid H_1')$, $\vec s'\leftarrow \tilde{\vec s}_{[(1-\beta)n-k]}$
	\For{$i=1$ \textbf{to} $\mathrm{poly}(n)\cdot S$}
	\State solve instance $(\tilde{H}',\vec s',\omega -p )$ via quantum algorithm returning $\vec e' \in \mathbb{F}_2^{(1-\beta)n}$ \label{line:RR-quantum}
	\State $\vec e''\leftarrow H_2'\vec e_{[n-k-\beta n+1,(1-\beta)n]}'+\tilde{\vec s}_{[n-k-\beta n +1,n-k]}$
	\If {$\wt(\vec e'')\leq p$}
	\State $\vec e\leftarrow P(\vec e_{[n-k-\beta n]}',\vec e'',\vec e_{[n-k-\beta n+1,(1-\beta)n]}')$
	\State \bf{break}
	\EndIf
	\EndFor
	\Until{$H\vec e = \vec s$}
	\State \Return $\vec e$
	\end{algorithmic}
	\caption{Punctured Hybrid}
	\label{alg:punctured-hybrid}
    
	\end{algorithm}
	\begin{theorem}[Punctured Hybrid]
	Let $n\in\mathbb{N}$, $\omega\in [n]$ and $k=Rn$ for $R\in \llbracket0,1\rrbracket$. Then for any qubit reduction factor $\delta \in \llbracket0,1\rrbracket$ \cref{alg:punctured-hybrid} solves the $\SDP{n}{k}{\omega}$ problem in expected time $T_{\textsc{PH}}$ using $\delta (1-R)Rn^2$ qubits for the matrix representation, where 
	\[
	    T_{\textsc{PH}}=\tmo{\frac{\binom{n}{\omega}}{\sqrt{\binom{(1-\beta)n}{\omega-p}\binom{(1-\beta-R)n}{\omega-p}}\binom{\beta n}{p}}\cdot\max\left(1,\sqrt{\binom{(1-\beta)n}{\omega-p}\cdot2^{-(1-\beta-R)n}}\right)}
	\]
	with $\beta = (1-\delta)(1-R)$ and $p\in[\min(\omega,\beta n)]$. 
	\label{thm:punctured-hybrid}

	\end{theorem}
	\begin{proof}
	Assume that the permutation distributes the error weight, such that for $P^{-1}\vec e = (\vec e_1,\vec e_2,\vec{e_3})\in \F_2^{(1-\beta-R)n}\times\F_2^{\beta n}\times\F_2^{Rn}$ it holds $\wt(\vec e_2)=p$.
	Now consider the permuted parity-check matrix in systematic form $\tilde{H}$ as given in line \ref{line:RR-systematic} of \cref{alg:punctured-hybrid} with corresponding syndrome $\tilde{\vec s}$. We obtain
	\[
	\tilde{H}P^{-1}\vec e= (\vec e_1+H_1'\vec e_3,\vec e_2+H_2'\vec e_3)=\tilde{\vec s}.
	\]
	This implies that $(\vec e_1,\vec e_3)$ is a solution to the syndrome decoding instance $(\tilde{H}',\vec s',\omega -p)$ with $\tilde{H}'=(I_{(1-\beta-R)n}\mid H_1')$ and $\vec s'=\tilde{\vec s}_{[(1-\beta-R)n]}$.
	The solution is then recovered by the application of our quantum circuit in line \ref{line:RR-quantum}. Note that in expectation there exist
	\[
	S:=\binom{(1-\beta)n}{\omega-p}\cdot2^{-(1-\beta-R)n}
	\]
	solutions to our reduced instance. Since we apply our quantum circuit $\mathrm{poly}(n)\cdot S$ times and in each execution a random solution is returned, a standard coupon collector argument yields that we recover all $S$ solutions with high probability.
	Now, when $\vec e'=(\vec e_1,\vec e_3)$ is returned by the quantum circuit, we recover $\vec e_2=\tilde{\vec s}_{[(1-\beta-R)n+1,(1-R)n]}+H_2'\vec e_3$ and eventually return $\vec e=P(\vec e_1,\vec e_2,\vec e_3)$.
	
	Next let us consider the time complexity of the algorithm. Observe that the probability, that $\wt(\vec e_2)=p$ for a random permutation holds is
	\begin{align*}
	q_\mathrm{C}:=\p{\wt(\vec e_2)=p}=\frac{\binom{(1-\beta)n}{\omega-p}\binom{\beta n}{p}}{\binom{n}{\omega}}.
	\end{align*}
    Hence, after $q_\mathrm{C}^{-1}$ iterations we expect that there is at least one iteration where $\wt(\vec e_2)=p$. 
    In each iteration we apply our quantum circuit $\tmo{S}$ times to solve the reduced instance $(\tilde{H}',\vec s',\omega -p)$, corresponding to an  $[(1-\beta)n, (1-R)n]$-code. Since there exist $S$ solutions the expected time to retrieve one of them at random is 
    \[
    T_\mathrm{Q}=\tmo{\sqrt{\frac{\binom{(1-\beta)n}{\omega-p}}{\max(1,S)\cdot\binom{(1-\beta-R)n}{\omega -p}}}},
    \]
    according to \cref{rem:isd-quant}. The maximum follows since we know that there exists at least one solution. In summary the running time becomes $T_{\textsc{PH}}=q_\mathrm{C}^{-1}\cdot T_\mathrm{Q}\cdot \max(1,S)$,
    as stated in the theorem. 
    
    The required amount of qubits of the quantum circuit for solving the syndrome decoding problem related to an $[(1-\beta)n, (1-R)n]$-code are roughly $R(1-\beta-R)n^2$ (compare to \cref{eqn:reduced-qubits}). Thus, for $\beta:=(1-\delta)(1-R)$ this corresponds to a qubit reduction of
    \[
    \frac{R(1-\beta-R)}{R(1-R)}=\frac{1-R -(1-\delta)(1-R)}{1-R}=\delta.
    \qedeq
    \]
    
    \end{proof}
	
	\cref{thm:punctured-hybrid} allows to easily determine the corresponding speedup, whose exact formula we give in \cref{cor:punctured-hybrid} in \cref{app:punctured}. 
	
	In \cref{fig:punctured-hybrid-plots} we compare the behavior of our new trade-off to our previously obtained \HybridPrange. Recall that the performance of \HybridPrange{} is not very sensitive to changes in the weight. Thus, for settings with a rate of $R=0.5$ the dashed lines are almost on top of each other. The value $p$ of our new trade-off (\cref{thm:punctured-hybrid}) were optimized numerically. It can be observed, that our second trade-off outperforms the first one for all parameters. We observe the best behaviour for low coderates and small error weights, which correspond to the case, where the solution is very unique. In these cases our \PuncturedHybrid{} algorithm can disregard parity-check equations without introducing multiple solutions to the reduced instance. Hence, still a single execution of the quantum circuit suffices to recover the solution. Note that in the McEliece, BIKE and HQC setting the error weight is sublinear, which is in favour of our new trade-off. BIKE and HQC furthermore use a very small error weight of only $\mathcal{O}(\sqrt{n})$ and specify a rate of $R=0.5$, which results in a very unique solution. Consequently, in \cref{fig:punctured-hybrid-plots} it can be observed, that asymptotically for these settings the second trade-off improves drastically on \HybridPrange. 
	
	Note that our formulation of the speedup for \PuncturedHybrid{} in contrast to \HybridPrange{} (see \cref{cor:hybrid-prange}) still depends on the error-rate, not exactly allowing for $\omega = o(n)$. Thus, to obtain the asymptotic plot we compared the result of \cref{cor:hybrid-prange} to \cref{thm:punctured-hybrid} for McEliece $[6688, 5024, 128]$, BIKE $[81946, 40973, 264]$ and HQC $[115274, 57637, 262]$, which are the suggested parameters for 256-bit security from the corresponding NIST submission documentation \cite{chou2020classic,aragon2017bike,melchor2020hamming}. 
	
	To quantify the result of our new trade-off take e.g. the case of McEliece and a qubit reduction by 80\% ($\delta=0.2$), as before. Here we improve to a speedup of $t(\delta)\approx 0.74$, compared to $0.82$ for $\HybridPrange$.

	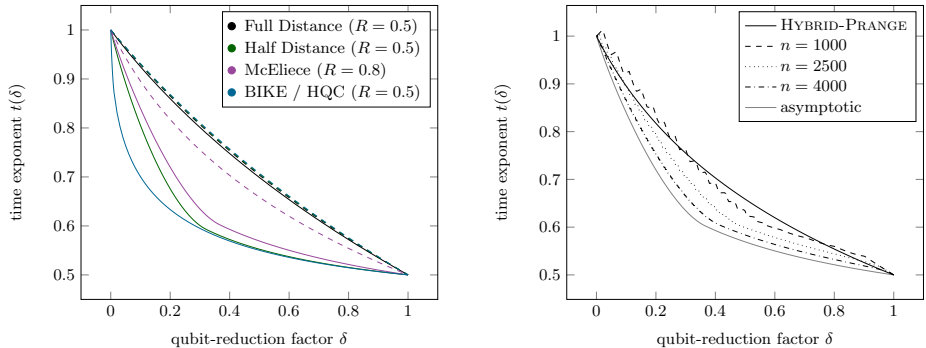
\begin{figure}[ht]
\centering
\begin{subfigure}[t]{.475\textwidth}
  \centering
	\resizebox{\textwidth}{!}{	\begin{tikzpicture}
	\begin{axis}[
		y tick label style={
			/pgf/number format/.cd,
			fixed,
			precision=3,
			/tikz/.cd
		},
		x tick label style={
			/pgf/number format/.cd,
			fixed,
			1000 sep={},
			precision=2,
			/tikz/.cd
		},
		ytick={0.5,0.6,0.7,0.8,0.9,1},
		xlabel={qubit-reduction factor $\delta$},
		ylabel={time exponent $t(\delta)$},
		legend cell align={left},
		legend entries={~Full Distance ($R=0.5$), ~Half Distance ($R=0.5$), ~McEliece ($R=0.8$),~BIKE / HQC ($R=0.5$)}
		]
		\addlegendimage{only marks,black}
		\addlegendimage{only marks,DarkGreen}
		\addlegendimage{only marks,Purple}
		\addlegendimage{only marks,MidnightBlue}
		
		\pgfplotstableread{plots/first_full_distance_asym.txt}
		\cTradeOff
		\addplot[color=black,dashed] table[x = X,y=Y] from \cTradeOff ;
		
		\pgfplotstableread{plots/second_full_distance_asym.txt}
		\cTradeOff
		\addplot[color=black] table[x = X,y=Y] from \cTradeOff ;
		
		\pgfplotstableread{plots/first_half_distance_asym.txt}
		\cTradeOff
		\addplot[color=DarkGreen,dashed] table[x = X,y=Y] from \cTradeOff ;
		
		\pgfplotstableread{plots/second_half_distance_asym.txt}
		\cTradeOff
		\addplot[color=DarkGreen] table[x = X,y=Y] from \cTradeOff ;
		
		\pgfplotstableread{plots/first_r=08_small_w.txt}
		\cTradeOff
		\addplot[color=Purple,dashed] table[x = X,y=Y] from \cTradeOff ;
		
		\pgfplotstableread{plots/second_mceliece_asym.txt}
		\cTradeOff
		\addplot[color=Purple] table[x = X,y=Y] from \cTradeOff ;
		
		\pgfplotstableread{plots/first_bike_asym.txt}
		\cTradeOff
		\addplot[color=MidnightBlue,dashed] table[x = X,y=Y] from \cTradeOff ;
		
		\pgfplotstableread{plots/second_bike_asym.txt}
		\cTradeOff
		\addplot[color=MidnightBlue] table[x = X,y=Y] from \cTradeOff ;
	\end{axis}	
\end{tikzpicture}}
	\caption{Asymptotic time exponents. New \cref{thm:punctured-hybrid} depicted as solid line, \cref{cor:hybrid-prange} as dashed line.}
	\label{fig:punctured-hybrid-plots}
\end{subfigure}%
\hspace{0.04\textwidth}
\begin{subfigure}[t]{.475\textwidth}
	    \centering
    	\resizebox{\textwidth}{!}{	\begin{tikzpicture}
	\begin{axis}[
		y tick label style={
			/pgf/number format/.cd,
			fixed,
			precision=3,
			/tikz/.cd
		},
		x tick label style={
			/pgf/number format/.cd,
			fixed,
			1000 sep={},
			precision=2,
			/tikz/.cd
		},
		ytick={0.5,0.6,0.7,0.8,0.9,1},
		xlabel={qubit-reduction factor $\delta$},
		ylabel={time exponent $t(\delta)$},
		legend cell align={left},
		legend entries={\HybridPrange, $n=1000$,$n=2500$,$n=4000$,asymptotic}
		]


		\pgfplotstableread{plots/first_R=08_small_w.txt}
		\cTradeOff
		\addplot[color=black] table[x = X,y=Y] from \cTradeOff ;

		\pgfplotstableread{plots/second_concrete_1000.txt}
		\cTradeOff
		\addplot[dashed] table[x = X,y=Y] from \cTradeOff ;
		
		\pgfplotstableread{plots/second_concrete_2500.txt}
		\cTradeOff
		\addplot[dotted] table[x = X,y=Y] from \cTradeOff ;
		
		\pgfplotstableread{plots/second_concrete_4000.txt}
		\cTradeOff
		\addplot[dashdotted] table[x = X,y=Y] from \cTradeOff ;
		
		\pgfplotstableread{plots/second_mceliece_asym.txt}
		\cTradeOff
		\addplot[opacity=0.5] table[x = X,y=Y] from \cTradeOff ;

	\end{axis}	
	
\end{tikzpicture}}
    	\caption{Time exponents for concrete parameter sets. McEliece parameter sets satisfy $k=0.8n$ and $\omega=\left\lfloor\frac{n}{5\log n}\right\rfloor$.}
    	\label{fig:first-v-second}
\end{subfigure}

\caption{Comparison of time exponents of \HybridPrange{} and \PuncturedHybrid{} (y-axis) plotted as a function of the qubit-reduction factor $\delta$ (x-axis).}
\end{figure}
%
	However, for concrete medium sized parameters this asymptotic behaviour is not necessarily obtained. In \cref{fig:first-v-second} we therefore show a comparison of both trade-offs for concrete McEliece parameter sets. Here we furthermore used the more accurate time complexity formula involving binomial coefficients rather than its asymptotic approximation to compute the speedup $t(\delta)$. Note that the discontinuity for our new trade-off in these cases is due to the limitation to discrete choices of $p$. We find that for parameters up to $n \approx 2500$ both trade-offs remain superior to each other for certain reduction factors $\delta$. For larger values of $n$ the \PuncturedHybrid{} algorithm becomes favourable for all $\delta$.
	In the BIKE and HQC settings the \PuncturedHybrid{} algorithm is favourable already for small parameters corresponding to $n=1000$.
	
\subsection{Combining both trade-offs}

Next we show how to combine both previous trade-offs to achieve an improved version. Therefore we first reduce the code length and dimension, again by guessing $\alpha n$ zero coordinates of $\vec e$ and removing the corresponding columns form $H$. The remaining instance is then solved using our \PuncturedHybrid{} algorithm (compare also to \cref{fig:combined-tradeoff}). If the initial guess was wrong, this procedure will not finish. Thus, we introduce an abort of the execution after the expected amount of iterations of \PuncturedHybrid{} on a correct guess.

	\begin{figure}[ht]
	\centering
	\scalebox{0.6}{\begin{tikzpicture}

\node at (-5.25,-2.4) {\huge$\mathbf{e}$};
\node at (6.2,2.45) {\huge$\mathbf{s}$};

\draw[opacity=0.5,densely dotted] (-4,2) -- (0.5,-1.65);
\node at (-2.25,0.75) {\huge$I_{n-k-\beta n}$};
\draw[opacity=0.5,dashed] (0.5,2) node (v1) {} -- (0.5,-1.65);

\node at (2.35,0.75) {\huge$H'$};
\fill[pattern=north west lines]  (4.05,-2.15) rectangle (0.5,-2.65);
\fill[color=white,opacity=0.5]  (4.05,-2.15) rectangle (0.5,-2.65);
\draw (0.5,-2.15) -- (0.5,-2.65);

\fill[pattern=north west lines]  (4.05,-2.15) rectangle (0.5,-2.65);
\fill[color=white,opacity=0.5]  (4.05,-2.15) rectangle (0.5,-2.65);
\draw (0.5,-2.15) -- (0.5,-2.65);

\fill[pattern=north west lines]  (-4,-2.15) rectangle (-0.8,-2.65);
\fill[color=white,opacity=0.5]  (-4,-2.15) rectangle (-0.8,-2.65);
\draw (-0.8,-2.15) -- (-0.8,-2.65);

\fill[opacity=0.5,pattern=crosshatch]  (4.05,2)  rectangle (5.45,-1.65);
\fill[opacity=0.65,color=white]  (4.05,2) node (v3) {} rectangle (5.45,-1.65);

\draw  (-4,-2.15) rectangle (5.45,-2.65);
\draw  (5.95,2) node (v2) {} rectangle (6.45,-1.65);
\node at (-5.25,0.5) {\huge$H$};

\fill[opacity=0.5,pattern=crosshatch]  (5.95,-0.5)  rectangle (6.45,-1.65);
\fill[opacity=0.65,color=white]  (5.95,-0.5) rectangle (6.45,-1.65);
\draw[dashed] (6.45,-0.5) -- (5.95,-0.5);

\fill[opacity=0.5,pattern=crosshatch]  (-4,-0.5)  rectangle (4.05,-1.65);
\fill[opacity=0.65,color=white]  (-4,-0.5) rectangle (4.05,-1.65);
\draw[dashed] (4.05,-0.5) -- (-4,-0.5);

\node at (-0.15,-3) {$\beta n$};

\draw (-0.8,-2.9) -- (-0.8,-3.1);
\draw (-0.8,-3) -- (-0.5,-3);
\draw (0.2,-3) -- (0.5,-3);
\draw (0.5,-2.9) -- (0.5,-3.1);

\draw (-4.45,-0.75) -- (-4.45,-0.5);
\draw (-4.55,-0.5) -- (-4.35,-0.5);
\draw (-4.45,-1.45) -- (-4.45,-1.65);
\draw (-4.35,-1.65) -- (-4.55,-1.65);

\node at (4.8,-3) {$\alpha n$};

\draw (4.05,-2.9) -- (4.05,-3.1);
\draw (4.05,-3) -- (4.45,-3);
\draw (5.15,-3) -- (5.45,-3);
\draw (5.45,-2.9) -- (5.45,-3.1);

\node[rotate=90] at (-4.45,-1.1) {$\beta n$};
\draw[opacity=0.5,dashed] (-0.8,2) -- (-0.8,-0.5);
\node at (-0.15,0.75) {\huge$\mathbf{0}$};
\node at (-0.1,-2.4) {$p$};
\draw[fill opacity=0.04,thick]  (v1) rectangle (4.05,-0.5);
\draw[thick]  (v2) rectangle (6.45,-0.5);

\draw  (-4,2) rectangle (5.45,-1.65);
\draw (4.05,-2.15) -- (4.05,-2.65);
\node at (4.8,-2.4) {\Large$\mathbf{0}$};
\draw  (v3) rectangle (v3);
\draw[dashed] (4.05,-0.5) -- (4.05,-1.65);
\end{tikzpicture}}

	\caption{Input matrix in systematic form where $\beta n$ parity-check equations are omitted and $\alpha n$ zeros of $\vec e$ are known. The vector $\vec e$ is assumed to contain weight $p$ on $\beta n$ coordinates. Framed parts are used as input to the quantum algorithm.}
	\label{fig:combined-tradeoff}

	\end{figure}
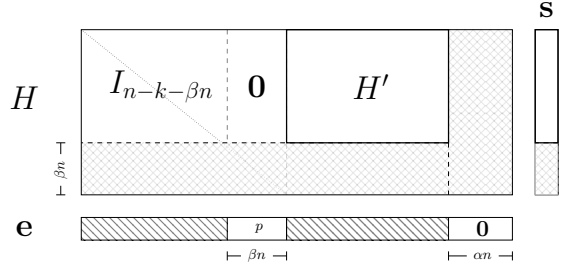
    The pseudocode of the procedure is given in \cref{alg:combined-tradeoff}. Note that here we use $\beta$ and $p$ as input parameters to \PuncturedHybrid, rather than to the choice made in \cref{alg:punctured-hybrid} (\PuncturedHybrid). 
	\begin{algorithm}
	\begin{algorithmic}[1]
		\Require{parity-check matrix $H\in\mathbb{F}_2^{(n-k)\times n}$, syndrome $\vec s \in \mathbb{F}_2^{n-k}$, weight $\omega\in[n]$, qubit reduction factor $\delta\in \llbracket0,1\rrbracket$}
		\Ensure{error vector $\vec e$ with $\wt(\vec e)=\omega$ satisfying $H\vec e=\vec s$ }
		\State choose $\alpha$ and $p$ accordingly
		\State $\beta:= (1-\frac{k}{n}) \left(\frac{\delta \frac{k}{n}}{\frac{k}{n}-\alpha}\right)$, $E:=\frac{\binom{(1-\alpha)n}{\omega}}{\binom{(1-\alpha-\beta)n}{\omega -p}\binom{\beta n}{p}}$
		\Repeat
		\State choose random permutation matrix $P\in\mathbb{F}_2^{n\times n}$ and set $\tilde{H}\leftarrow HP$
		\State $\vec e' \leftarrow$ \PuncturedHybrid$(\tilde{H}_{[(1-\alpha)n]},\vec s,\omega,\delta,\frac{\beta}{1-\alpha},p)$  \Comment{abort after $E$ iterations of the outer loop}
		\State $\vec e \leftarrow P(\vec e',0^{\alpha n})$ 
		\Until{$H\vec e = \vec s$}
		\State\Return $\vec e$
	\end{algorithmic}
	\caption{\Combination}
	\label{alg:combined-tradeoff}
	
\end{algorithm}

	\begin{restatable}[Combined Hybrid]{theorem}{combinedthm}
	Let $n\in\mathbb{N}$, $\omega\in [n]$ and $k=Rn$ for $R\in \llbracket0,1\rrbracket$. Then for any qubit reduction factor $\delta \in \llbracket0,1\rrbracket$ the $\SDP{n}{k}{\omega}$ problem can be solved in expected time $T_{\textsc{CH}}$ using $\delta (1-R)Rn^2$ qubits for the matrix representation, where 
	{\small
	\begin{align*}
	    T_{\textsc{CH}}&=\tilde{\mathcal{O}}\Bigg(\frac{\binom{n}{\omega}}{\sqrt{\binom{(1-\alpha-\beta)n}{\omega-p}\binom{(1-\beta-R)n}{\omega-p}}\binom{\beta n}{p}}
	    \cdot\max\left(1,\sqrt{\binom{(1-\alpha-\beta)n}{\omega-p}\cdot2^{-(1-\beta-R)n}}\right)\Bigg)
	\end{align*}
	}
	with $\alpha\in\llbracket0,R\rrbracket$, $\beta = (1-R) \left(1-\frac{\delta R}{R-\alpha}\right)$ and $p\in[\min(\omega,\beta n)]$. 
	\label{thm:combined-tradeoff}
	\end{restatable}
	\begin{proof}
	The correctness follows from the correctness of \cref{alg:hybrid-prange} and \cref{alg:punctured-hybrid}. Therefore observe that for a correct guess of $\alpha n$ zero positions of $\vec e$, the expected amount of permutations needed by \PuncturedHybrid{} to find the solution is
	\[
	E:=\frac{\binom{(1-\alpha)n}{\omega}}{\binom{(1-\alpha-\beta)n}{\omega -p}\binom{\beta n}{p}}.
	\]
	Also note that \PuncturedHybrid{} is called on a code of length $n'=(1-\alpha)n$. Hence setting $\beta'=\frac{\beta}{1-\alpha}$ guarantees that $\beta'n'=\beta n$ parity equations are omitted.
	
	For the time complexity we have again with probability 
	\[
	q_\mathrm{C}:=\mathrm{Pr}\left[P^{-1}\vec e=(\vec e_1,0^{\alpha n})\right]=\frac{\binom{(1-\alpha)n}{\omega}}{\binom{n}{\omega}},
	\]
	a correct guess for $\alpha n$ zero positions (compare to the proof of \cref{thm:first-tradeoff}). 
	In each iteration of our combined algorithm we call the \PuncturedHybrid{} algorithm. Inside this subroutine $E$ iterations of the outer loop are executed, each performing 
	\[
	S=\tmt{\max\left(1,\frac{\binom{1-\beta-\alpha}{\omega-p}}{2^{-(1-R-\beta)n}}\right)}
	\]
	calls to the quantum circuit. This quantum circuit is applied to solve the syndrome decoding problem defined on an $[(1-\alpha-\beta)n,(R-\alpha)n]$-code with error-weight ${\omega-p}$ (compare to \cref{fig:combined-tradeoff}), which takes time
	\[
	    T_\mathrm{Q}=\tmo{\sqrt{\frac{\binom{(1-\alpha-\beta)n}{\omega-p}}{S\cdot\binom{(1-\beta-R)n}{\omega-p}}}}.
	\]
	Thus, eventually, the time complexity of the whole algorithm summarizes as $T_{\textsc{CH}}=q_\mathrm{C}^{-1}\cdot E\cdot T_\mathrm{Q}\cdot S$, as claimed. Finally, note that for given $\beta =  (1-R) \left(1-\frac{\delta R}{R-\alpha}\right)$ we obtain a qubit reduction by
	\[
	\frac{(R-\alpha)(1-R-\beta)}{R(1-R)}=\frac{(R-\alpha)(1-R)(1-(1-\frac{\delta R}{R-\alpha})}{R(1-R)}=\delta.
	\qedeq
	\]
\end{proof}

Next we give a comparison of the trade-off behavior in different settings. On the left in \cref{fig:combined-tradeoff-plots} we illustrate the asymptotic behaviors of the trade-offs, where $p$ and $\alpha$ for the combined trade-off were numerically optimized. It shows that the combination of both trade-offs (dashed lines) for most parameters improves on \PuncturedHybrid{} (solid line). Especially in the full distance decoding setting an improvement for nearly all $\delta$ is achieved. This is due to the fact, that the guessing of zero coordinates is an additional possibility to control the amount of solutions to the reduced instance and therefore to optimize the complexity of the \PuncturedHybrid{} subroutine. This is also the reason why we achieve no (asymptotic) improvement in the BIKE and HQC settings, here the solution is already so unique that the trade-off can not benefit from the new degree of freedom.

\begin{figure}[ht]
\centering
\begin{subfigure}[t]{.475\textwidth}
  \centering
	\resizebox{\textwidth}{!}{	\begin{tikzpicture}
	\begin{axis}[
		y tick label style={
			/pgf/number format/.cd,
			fixed,
			precision=3,
			/tikz/.cd
		},
		x tick label style={
			/pgf/number format/.cd,
			fixed,
			1000 sep={},
			precision=2,
			/tikz/.cd
		},
		ytick={0.5,0.6,0.7,0.8,0.9,1},
		xlabel={qubit-reduction factor $\delta$},
		ylabel={time exponent $t(\delta)$},
		legend cell align={left},
		legend entries={~Full Distance ($R=0.5$), ~Half Distance ($R=0.5$), ~McEliece ($R=0.8$),~BIKE / HQC ($R=0.5$)}
		]
		\addlegendimage{only marks,black}
		\addlegendimage{only marks,DarkGreen}
		\addlegendimage{only marks,Purple}
		\addlegendimage{only marks,MidnightBlue}
		

		\pgfplotstableread{plots/second_full_distance_asym.txt}
		\cTradeOff
		\addplot[color=black] table[x = X,y=Y] from \cTradeOff ;
		
		\pgfplotstableread{plots/combined_full_distance_asym.txt}
		\cTradeOff
		\addplot[color=black, dashed] table[x = X,y=Y] from \cTradeOff ;

		\pgfplotstableread{plots/second_half_distance_asym.txt}
		\cTradeOff
		\addplot[color=DarkGreen] table[x = X,y=Y] from \cTradeOff ;
		
		\pgfplotstableread{plots/combined_half_distance_asym.txt}
		\cTradeOff
		\addplot[color=DarkGreen,dashed] table[x = X,y=Y] from \cTradeOff ;
		

		\pgfplotstableread{plots/second_mceliece_asym.txt}
		\cTradeOff
		\addplot[color=Purple] table[x = X,y=Y] from \cTradeOff ;
		
		\pgfplotstableread{plots/combined_mceliece_asym.txt}
		\cTradeOff
		\addplot[color=Purple, dashed] table[x = X,y=Y] from \cTradeOff ;
		
		\pgfplotstableread{plots/second_bike_asym.txt}
		\cTradeOff
		\addplot[color=MidnightBlue] table[x = X,y=Y] from \cTradeOff ;
	\end{axis}	\end{tikzpicture}}
\end{subfigure}%
\hspace{0.04\textwidth}
\begin{subfigure}[t]{.475\textwidth}
	    \centering
    	\resizebox{\textwidth}{!}{	\begin{tikzpicture}
	\begin{axis}[
		y tick label style={
			/pgf/number format/.cd,
			fixed,
			precision=3,
			/tikz/.cd
		},
		x tick label style={
			/pgf/number format/.cd,
			fixed,
			1000 sep={},
			precision=2,
			/tikz/.cd
		},
		ytick={0.5,0.6,0.7,0.8,0.9,1},
		xlabel={qubit-reduction factor $\delta$},
		ylabel={time exponent $t(\delta)$},
		legend cell align={left},
		legend entries={\HybridPrange, \PuncturedHybrid,\Combination,asymptotic}
		]


		\pgfplotstableread{plots/first_R=08_small_w.txt}
		\cTradeOff
	    \addplot[color=black] table[x = X,y=Y] from \cTradeOff ;
	    
	    \pgfplotstableread{plots/second_concrete_2500.txt}
		\cTradeOff
		\addplot[dotted] table[x = X,y=Y] from \cTradeOff ;
		
	    \pgfplotstableread{plots/combined_concrete_2500.txt}
		\cTradeOff
		\addplot[dashdotted] table[x = X,y=Y] from \cTradeOff ;
		
        \pgfplotstableread{plots/combined_mceliece_asym.txt}
		\cTradeOff
		\addplot[opacity=0.5] table[x = X,y=Y] from \cTradeOff ;
	\end{axis}	
	
\end{tikzpicture}}
    
\end{subfigure}

\caption{Comparison of time exponents (y-axis) plotted as a function of the qubit-reduction factor $\delta$ (x-axis) for asymptotic (left, combined trade-off illustrated dashed, \cref{thm:punctured-hybrid} solid) and concrete McEliece parameters (right).}
\label{fig:combined-tradeoff-plots}
\end{figure}
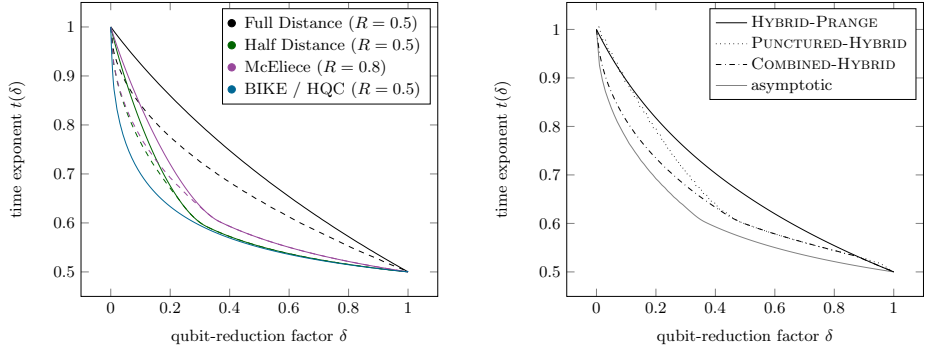
But also in the McEliece setting we achieve notable improvements. If we again consider a reduction-factor of $\delta=0.2$ the combination improves the speedup to $t(\delta)\approx0.69$ from 0.74 achieved by \PuncturedHybrid.
Furthermore, when focusing on near future realizations, i.e., the regime of small reduction factors, it is for example possible with just one percent of the qubits ($\delta=0.01$) to achieve a speedup of $t(\delta)\approx0.92$. 

On the right in \cref{fig:combined-tradeoff-plots} we show the relation between qubit reduction and speedup for concrete McEliece parameters. Here we restrict ourselves to the instance $n=2500$ for the sake of clarity. But note that for all parameter sets at least the minimum of both trade-offs is obtained with improvements especially for low reduction factors. In the BIKE and HQC setting for small parameter sets with $n\leq3000$ we achieve (small) improvements in the regime of $\delta\leq 0.05$.

\subsection{Overview and Discussion}

For convenience we state in \cref{tab:quantum-instances} the parameters of the reduced instances solved by the quantum circuits within each of our trade-offs. \cref{tab:circuit-complexities} then states the necessary amount of qubits and the resulting depth of the circuits to solve a respective instance with parameters $(n,k,\omega)$. 

\begin{table}[]
\centering
\begin{tabular}{l||c|c|c}
                                  & \multicolumn{3}{c|}{Quantum Instance Parameters} \\
                                  & \multicolumn{1}{c|}{$n'$} & \multicolumn{1}{c|}{$k'$} & \multicolumn{1}{c|}{$\omega'$} \\ \hline\hline
\HybridPrange      & $(1-\alpha)n$             & $(R-\alpha)n$             & $\omega$                        \\
\PuncturedHybrid & $(1-\beta)n$              & $Rn$                      & $\omega-p$                     \\
\Combination       & $(1-\alpha-\beta)n$       & $(R-\alpha)n$             & $\omega-p$                    
\end{tabular}
\caption{Parameters of the reduced subinstance solved by the quantum circuit called by the respective classical co-processor. }
\label{tab:quantum-instances}
\end{table}

By plugging in the values from \cref{tab:quantum-instances} into the formulas given in \cref{tab:circuit-complexities} one receives the quantum complexities of the respective classical co-processor. Here we differentiate between optimization regarding the amount of qubits and the circuit depth. The essential difference lies in the use of the Lee-Brickell improvement in case of an optimization of the depth, while the width optimized variant uses the qubit reduction technique from \cref{sec:width-improvement}.

The depth of our circuits is mainly dominated by the application of the Gaussian elimination, where the additional $\log n$ factor results from the decomposition of multi-controlled gates. Remember that $q$ is the proportion of subsets yielding a solution among all size-$(n-k)$ subsets in the case of Prange, given in \cref{eqn:prob-isd}, while $q_\textrm{LB}$ is the proportion of good subsets in the case of the Lee-Brickell algorithm stated in \cref{eqn:prob-lee-brickell}.




\begin{table}
\centering
\begin{tabular}{c||c|c|c}
\multicolumn{1}{l||}{} & $W$-Optimized  & $D$-Optimized &  \multicolumn{1}{c|}{$D$-Optimized (cyclic)}  \\ \hline\hline\\[-0.3cm]

\begin{tabular}[c]{@{}c@{}}Qubit\end{tabular} & \begin{tabular}[c]{@{}c@{}}$(n-k+2)(k+3)-7$\end{tabular} & $(n-k+1)\cdot (n+2) -3$ & $(n-k)\cdot (n+k+2) -1$  \\[0.15cm]
\begin{tabular}[c]{@{}c@{}}Depth\end{tabular}  &  $\mathcal{O}\left(\frac{n^3\log n}{\sqrt{q}}\right)$ & $\mathcal{O}\left(\frac{n^3\log n+p\binom{k}{p}}{\sqrt{q_\textrm{LB}}}\right)$ & $\mathcal{O}\left(\frac{n\left(n^2\log n+p\binom{k}{p}\right)}{\sqrt{k\cdot q_\textrm{LB}}}\right)$ \\

\end{tabular}

\caption{Depth and required qubits of the quantum circuit to solve the $\SDP{n}{k}{\omega}$. }
\label{tab:circuit-complexities}
\end{table}


Overall, we presented concrete depth and width optimized quantum circuits for the fully-fledged ISD procedure. Our tradeoffs put a special focus on the reduction of necessary qubits, targeting \emph{near-term} realizations. Following this thought, we also provide the necessary implementations of our circuits in the simulation library Qibo making a transition to a real quantum computer as easy as possible.

Although we placed a strong focus on the circuit width, we have shown that ISD  can also be implemented efficiently on a quantum computer from a depth perspective. Thus, we cannot confirm the mentioned statement of the McEliece submission regarding a higher overhead when applying a Grover search to ISD rather than AES. However, we admit that a single application of AES has a lower complexity than one iteration of an ISD algorithm, which lies in favor of the quantum security of code-based schemes. Since NIST imposes a depth-limitation on the used quantum circuits, the more depth is needed for the implementation of one iteration, the less Grover iterations can be performed. We made a first step in the direction of overcoming this issue by giving a quantized version of the Lee-Brickell improvement and by exploiting the cyclicity in the BIKE / HQC cases. Both approaches tackle the problem by reducing the number of necessary Grover iterations. The second possibility is targeting a depth reduction of a single iteration, which is dominated by performing the Gaussian elimination. We leave it as an open problem to further study the concrete quantum circuit design of advanced Gaussian elimination procedures, such as M4RI or Strassen.


\renewcommand\appendix{\par
  \setcounter{chapter}{0}%
  \setcounter{section}{0}%
  \def\@chapapp{Appendix}%
  \def\thechapter{\@Alph\c@chapter}%
  \renewcommand\thesection{\Roman{section}}}

\bibliographystyle{splncs04}
\bibliography{Citations}

\appendix
\crefalias{section}{supp}

\renewcommand\chaptername{Appendix}
\def\thechapter{S}

\chapter*{Appendix}


	
\section{Width reduced circuit}
\label{app:width-reduce}
\cref{alg:adapted-combined-circuit} shows the pseudocode for generating our width optimized circuit, not requiring the idenity part of $H$ as an input, described in \cref{sec:width-improvement}.

	\begin{algorithm}[h]
		\begin{algorithmic}[1]
		\Require{matrix $H'\in\F_2^{(n-k)\times k}$, syndrome $\vec s\in \F_2^{n-k}$, $(k+1)(n-k)+n+\lceil\log (n-k)\rceil  $ qubits}
		\Ensure{Uniform superposition over weight of all $\vec x$ with $H_I\vec x=\vec s$ for $I\in\binom{[n]}{n-k}$ where $H=(I_{n-k}\mid H')$}
		\State Initialize qubits with $(H'\mid \vec s)$
		\State Generate uniform superposition over $\Vset{n}{n-k}$ on qubits $(b_1,\ldots,b_n)$
		\State $c\leftarrow n-k$ 
			   \vspace{0.2cm}
		\For{$i = 1$ \textbf{to} $n-k$} \Comment{row swaps depending on first $n-k$ columns}
		\For{$j = 1$ \textbf{to} $i-1$}
		\If{$b_i=1$ \textbf{and} $c=n-k-j+1$}{ swap row $i$ and row $j$}
		\EndIf
		\EndFor
		\IfThen {$b_i=1$}{$c\leftarrow c-1$}
	    \EndFor
	    \vspace{0.2cm}
	    
		\For{$i = n$ \textbf{down to} $n-k+1$} \Comment{swap selected columns of $H'$ to the back}
		    \If{$b_i = 1$}
		        \State $c\leftarrow c-1$
		        \State $x_i\leftarrow 1$
        		\For{$j = i$ \textbf{to} $n-1$}
        		    \State swap column $j$ and $j+1$
        		    \State swap $x_j$ and $x_{j+1}$
        		\EndFor
        	\EndIf
        \EndFor
	    \vspace{0.2cm}
		\State Apply Gaussian elimination circuit starting with the last column, where each operation depending on column $j$ is controlled by $x_j$ (\cref{alg:gaussian})
		\State \Return $c \leftarrow$ weight of last column  
		\end{algorithmic}
		\caption{Width-reduced Combined Circuit}
		\label{alg:adapted-combined-circuit}
	\end{algorithm}
\section{Punctured Hybrid}
\label{app:punctured}
In the following corollary we state the exact form of the speedup $t(\delta)$ for our \PuncturedHybrid{} (\cref{thm:punctured-hybrid}).

\begin{corollary}[Punctured Hybrid Speedup]
	Let $n\in\mathbb{N}$, $\omega=\tau n$ and $k=Rn$ , $p=\rho n$ for $\tau,R,\rho \in \llbracket0,1\rrbracket$ and let $T_\mathrm{C}$ be as defined in \cref{eqn:classical-complexity}. Then for any qubit reduction factor $\delta \in \llbracket0,1\rrbracket$ \cref{alg:punctured-hybrid} solves the $\SDP{n}{k}{\omega}$ problem in time $(T_\mathrm{C})^{t(\delta)}$ using $\delta (1-R)Rn^2$ qubits for the matrix representation, where 
	\[
	t(\delta)=\frac{\Hf(\tau)-\beta\Hfv{\frac{\rho}{\beta}}-\frac{1-\beta}{2}\cdot\Hfv{\frac{\tau-\rho}{1-\beta}}-\frac{(1-\beta-R)}{2}\cdot\Hfv{\frac{\tau-\rho}{1-\beta-R}}+\max(0,\sigma)}{\Hf(\tau)-(1-R)\Hfv{\frac{\tau}{1-R}}}
	\]
	for $\beta = (1-\delta)(1-R)$ and $\sigma=(1-\beta)\Hfv{\frac{\tau-\rho}{1-\beta}}-(1-\beta-R)$.
	\label{cor:punctured-hybrid}
	\end{corollary}
	\begin{proof}
	Recall that $t(\delta)=\frac{\log T_{\textsc{PH}}}{\log T_\mathrm{C}}$, where $T_{\textsc{PH}}$ is the running time of \cref{alg:punctured-hybrid}, given in \cref{thm:punctured-hybrid}. Now the statement of the corollary follows immediately by approximating the binomial coefficients in $T_{\textsc{PH}}$ and $T_\mathrm{C}$ via Stirling's formula (see \cref{eqn:stirling}).\qed 
	\end{proof}

\end{document}